\documentclass[11pt]{article}
\usepackage[margin=1in]{geometry}
\usepackage[latin2]{inputenc}
\usepackage[english]{babel}
\usepackage{amsmath}
\usepackage{amsthm}
\usepackage{amssymb}
\usepackage{microtype}
\usepackage{xspace}
\usepackage{comment}
\usepackage{letltxmacro}
\usepackage{comment}
\usepackage{enumerate}
\usepackage{tikz}
\usepackage{eufrak}

\usepackage{todonotes}
\presetkeys{todonotes}{inline,caption={}}{} 

\mathcode`l="8000
\begingroup
\makeatletter
\lccode`\~=`\l
\DeclareMathSymbol{\lsb@l}{\mathalpha}{letters}{`l}
\lowercase{\gdef~{\ifnum\the\mathgroup=\m@ne \ell \else \lsb@l \fi}}%
\endgroup

\newcommand{\h}[1]{\end{document}}

\newtheorem{define}{Definition}
\usepackage{boxedminipage}
 
\newcommand{\poly}{{\rm{poly}}}

\newcommand{\cS}{{\mathcal{S}}}
\newcommand{\cO}{{\mathcal{O}}}

\newcommand{\ch}{\mathrm{char}}

\usepackage{bold-extra}
\usepackage{lmodern}
\usepackage[T1]{fontenc}
\usepackage{textcomp}

\newcommand{\I}{\cal{I}}
\newcommand{\mat}{$M=(E,{\cal I})$}
\newcommand{\matl}[1]{$M_{#1}=(E_{#1},{\cal I}_{#1})$}

\newcommand{\whnd}[1]{\widehat{#1}}

\newcommand{\rep}[2] {$\widehat{{\cal #1}} \subseteq_{rep}^{#2} {\cal #1}$}

\newcommand{\bnoml}[2]{  $\binom{{#1}}{{#2}}$}

\newcommand{\rank}[1]{$\mbox{\sf rank}(#1)$}

\newcommand{\tgem}{$\cO\left({p+q \choose p} t p^\omega + t {p+q \choose q} ^{\omega-1} \right)$}

\newcommand{\repmat}[1]{$A_{#1}$}

\newcommand{\F}{{\mathbb F}}
\newcommand{\K}{{\mathbb K}}

\usepackage{graphicx}
\usepackage{caption}
\usepackage{subcaption}

\newtheorem{theorem}{Theorem}[section]
\newtheorem{lemma}[theorem]{Lemma}
\newtheorem{claim}[theorem]{Claim}
\newtheorem{corollary}[theorem]{Corollary}
\newtheorem{proposition}[theorem]{Proposition}

\theoremstyle{definition}
\newtheorem{definition}[theorem]{Definition}

\def\cqedsymbol{\ifmmode$\lrcorner$\else{\unskip\nobreak\hfil
\penalty50\hskip1em\null\nobreak\hfil$\lrcorner$
\parfillskip=0pt\finalhyphendemerits=0\endgraf}\fi}

\newcommand{\executeiffilenewer}[3]{%
\ifnum\pdfstrcmp{\pdffilemoddate{#1}}%
{\pdffilemoddate{#2}}>0%
{\immediate\write18{#3}}\fi%
} 
\title{Deterministic Truncation of Linear Matroids\thanks{%
D. Lokshtanov is supported by the BeHard grant under the recruitment programme of the of Bergen Research Foundation.
S. Saurabh is supported by PARAPPROX, ERC starting grant no. 306992.}}

\author{
Daniel Lokshtanov\thanks{University of Bergen, Norway. \texttt{daniello@ii.uib.no}}
\and {Pranabendu Misra\thanks{Institute of Mathematical Sciences, India. \texttt{\{pranabendu|fahad|saket\}@imsc.res.in}}}
\addtocounter{footnote}{-1}
\and{Fahad Panolan}\footnotemark \addtocounter{footnote}{-1}
\and {Saket Saurabh}\footnotemark \addtocounter{footnote}{-2} \footnotemark
}

\begin{document}

\date{}

\begin{titlepage}
\def\thepage{}
\thispagestyle{empty}
\maketitle
\begin{abstract}
Let \mat{} be a matroid. 
A {\em $k$-truncation} of $M$ is a matroid {$M'=(E,{\cal I}')$} such that for any $A\subseteq E$, 
$A\in {\cal I}'$ if and only if $|A|\leq k$ and $A\in \I$. 
Given a linear representation of $M$
we consider the problem of finding a linear representation of the $k$-truncation of this matroid.
This problem can be abstracted out to the following problem on matrices.
Let $M$ be a $n\times m$ matrix  over a field $\mathbb{F}$. 
A {\em rank $k$-truncation} of the matrix $M$ is a  $k\times m$   matrix $M_k$  
(over $\F$ or a related field) such that for every subset $I\subseteq \{1,\ldots,m\}$ of size at most $k$, the set of columns 
corresponding to $I$ in $M$ has rank  $|I|$ if and only if the corresponding set of columns in $M_k$ has rank $|I|$. 
Finding rank $k$-truncation of matrices is a common way to obtain a linear representation
of $k$-truncation of linear matroids, which has many algorithmic applications.
A common way to compute a rank $k$-truncation of a $n \times m$ matrix
is to multiply the matrix with a random  $k\times n$ matrix  (with the entries from a field of 
an appropriate size), yielding a simple randomized algorithm. 
So a natural question is   whether it possible to obtain a rank $k$-truncation of a matrix, {\em deterministically.}
In this paper we settle this question for matrices over any finite field or the field of rationals ($\mathbb Q$). 
We show that given a matrix $M$ over a field $\F$
we can compute a $k$-truncation $M_k$ over the ring $\F[X]$ in deterministic polynomial time. 

Our algorithms are based on the properties of the classical Wronskian determinant 
and the folded Wronskian determinant, which was recently introduced by Guruswami and Kopparty~[\,{\em FOCS, 2013}\,].
The folded wronskian determinant was originally defined using a primitive element of the field.  We prove that one can relax this 
to require only an element of ``polynomially large order''.  This then allows our algorithms to work over any finite field. 
This is our main technical contribution and we believe that it will be useful in other applications.


One of our motivations for studying $k$-truncation of matroids
stems from the efficient computation of ``representative families for linear matroids''.
As an application we show that, given a representation of any linear matroid, we can compute representative families deterministically.
This includes many important classes of matroids such as graphic matroids, co-graphic matroids, partition matroids and others.
Our result derandomizes many parameterized algorithms, including an algorithm  for computing {\sc $\ell$-Matroid Intersection}.

\end{abstract}
\end{titlepage}
\section{Introduction}\label{sec_intro}
%
A {\em rank $k$-truncation} of a $n \times m$ matrix $M$, is a  $k\times m$  matrix $M_k$  
such that for every subset $I\subseteq \{1,\ldots,m\}$ of size at most $k$, the set of columns 
corresponding to $I$ in $M$ has rank  $|I|$ if and only of the corresponding set of columns in $M_k$ has rank $|I|$. 
We can think of finding a rank $k$-truncation of a matrix as a dimension 
 reduction problem such that linear independence among all sets of columns of 
 size at most $k$ is preserved.  
This problem is a variant of the more general {\em dimensionality reduction} problem,
which is a basic problem in many areas of computer science such as
machine learning, data compression, information processing and others.
In dimensionality reduction, we are given a collection of points (vectors) in a high dimensional space,
 and the objective is to map these points to points in a space of small dimension while preserving some
 property of the original collection of points.
For an example, one could consider the problem of reducing the dimension of the space, while
preserving the pairwise distance, for a given collection of points.
Using the Johnson-Lindenstrauss Lemma  this can be done approximately for any collection of $m$ points,
 while reducing the dimension of the space to $\cO(\log m)$.

In this work, we study dimensionality reduction under the constraint that
linear independence of any subcollection of size up to $k$ of the given set of vectors is preserved. 
The motivation for this problem comes from {\em Matroid theory} and its algorithmic applications. 
For any matroid \mat{}, a {\em $k$-truncation} of $M$ is a matroid {$M'=(E,{\cal I}')$} such that for any $A\subseteq E$, 
$A\in I'$ if and only if $|A|\leq k$ and $A\in \I$. 
Given a linear representation of a matroid \mat{} of rank $n$ over a universe of size $m$ 
(which has a representation matrix $M$ of dimension $n \times m$),
the problem of finding a linear representation of the $k$-truncation of the matriod \mat{},
is precisely the dimensionality reduction problem on the matrix $M$.
Here the objective is to map the set of column vectors of $M$ (which lie in a space of dimension $n$)
to vectors in a space of dimension $k$ such that, any set $S$ of column vectors of $M$ with $|S|\leq k$ are linearly independent 
if and only if the corresponding set of vectors in the $k$-dimensional vector space are linearly independent.

A common way to obtain a rank $k$-truncation of a matrix $M$,
is to left-multiply $M$ by a random matrix of dimension $k\times n$ (with entries from a field of 
an appropriate size). Then using the  Schwartz-Zippel Lemma one can show that, 
the product matrix is a $k$-truncation of the matrix $M$ with high probability.
This raises a natural question of whether there is a deterministic algorithm for 
computing $k$-truncation of a matrix.
In this paper we settle this question by giving a polynomial time deterministic algorithm to solve this problem.
In particular we have the following theorem.
\begin{theorem} 
\label{main:thm: truncation}
Let $\F=\F_{p^l}$ be a finite field or  $\F=\mathbb{Q}$. 
Let $M$ be a $n\times m$ matrix over $\F$ of rank $n$. 
Given a number $k \leq n$, we can compute a matrix $M_k$ over the field $\F(X)$
such that it is a representation of the $k$-truncation of $M$.   
Furthermore, given $M_k$, we can test whether a given set of $l$ columns 
in $M_k$ are linearly  independent in $\cO(n^2k^3)$ field operations.   
\end{theorem}
\noindent 
We further show that for many fields,
the $k$-truncation matrix can be represented over a finite degree extension.


\paragraph{Tools and Techniques.}
The main tool used in this work, is the Wronskian determinant and its characterization of the linear independence of a set of 
polynomials.  Given a polynomial $P_j(X)$ and  a number $l$, define $Y_j^l=(P_j(X),P_j^{(1)}(X),\ldots,P_j^{(l-1)}(X))^T$. Here, $P_j^{(i)}(X)$ is the $i$-th formal derivative of $P(X)$. Formally, the Wronskian matrix of a  set of polynomials $P_1(X),\ldots,P_k(X)$ is defined as  the $k\times k$ matrix $W(P_1,\ldots,P_k)=[Y_1^k,\ldots,Y_k^k]$.  
Recall that to get a $k$-truncation of a linear matriod, we need to map a set of vectors from ${\mathbb F}^n$ to 
${\mathbb K}^k$ such that linear independence of any subset of the given vectors of size at most $k$ is preserved.  
We associate with each vector, a polynomial whose coefficients are the entries of the vector. A known mathematical result states that a set of polynomials $P_1(X),\ldots,P_k(X)\in \F[X]$ are linearly independent over $\F$  
if and only if the corresponding Wronskian determinant $\det(W(P_1,\ldots,P_k))\not \equiv 0$  in $\F[X]$~\cite{bostan2010wronskians, garcia1987wronskians, muir1882treatise}. 
However, this requires that the underlying field be  $\mathbb Q$ (or $\mathbb R$, $\mathbb C$),
or that it is a finite field whose characteristic is strictly larger than the maximum degree of $P_1(X),\ldots,P_k(X)$.

For fields of small characteristic, we use the notion of $\alpha$-folded Wronskian, 
which was introduced by Guruswami and Kopparty~\cite{GuruswamiK13} in the context of subspace design, with applications in coding theory.    Let $\F$ be a finite field  and $\alpha$ be an element of $\F$. 
Given a polynomial $P_j(X)\in \F[X]$ and a number $l$, define $Z_j^l=(P_j(X), P_j(\alpha X),\ldots,P_j(\alpha^{l-1}X))^T$.  
Formally, the $\alpha$-folded Wronskian matrix of a family of polynomials $P_1(X),\ldots,P_k(X)$ is defined as  the $k\times k$ matrix $W_\alpha(P_1,\ldots,P_k)=[Z_1^k,\ldots,Z_k^k]$.   Let $P_1(X),\ldots,P_k(X)$ be a family of polynomials of degree at most  $n-1$. Guruswami and Kopparty~\cite{GuruswamiK13} showed that if $\alpha$ is a primitive element of the field $\F$ and 
$|\F|>n$  then $P_1(X),\ldots,P_k(X)$ are linearly independent over $\F$ if and only if $\alpha$-folded Wronskian 
determinant
$\mathrm{det}(W_\alpha(P_1, \ldots, P_k))\not\equiv 0$ in $\F[X]$.  
However,  to use $\alpha$-folded Wronskians in algorithmic applications we need to find a primitive element of the finite field. 
Still, computing a primitive element of large finite field is a non-trivial problem and no deterministic polynomial time algorithms are known in general. 
To overcome this difficulty, we prove that one can relax the requirement on $\alpha$ such that it only needs to be an element of 
polynomially large order in $\F$. In particular we prove the following theorem.
\begin{theorem}
\label{main:thm: alpha folded wronskian}
Let $\F$ be a field, $\alpha$ be an element  of $\F$ of order $>(n-1)(k-1)$ and let 
$P_1(X),  \ldots, P_k(X)$ be a set of polynomials from $\F[X]$ of degree at most $n-1$. 
Then $P_1(X),  \ldots, P_k(X)$ are linearly independent over $\F$ if and only if the $\alpha$-folded Wronskian 
determinant 
$\mathrm{det}(W_\alpha(P_1, \ldots, P_k))\not\equiv 0$ in $\F[X]$.  
\end{theorem}

Given a $n\times m$ matrix $M$ over $\F$ and a positive integer $k$  our algorithm for finding a $k$-truncation of $M$ proceeds 
 as follows. To a column $C_i$ of  $M$  we associate a polynomial $P_i(X)$ whose coefficients are the entries of $C_i$. That is, if $C_i=(c_{1i},\ldots,c_{ni})^T$ then 
$P_i(X)= \sum_{j=1}^{n}  c_{ji}x^{j-1} .$  If the characteristic of the field $\F$ is strictly larger than $n$ or $\F=\mathbb{Q}$ then we return $M_k=[Y_1^k,\ldots,Y_m^k]$ as the desired $k$-truncation of $M$. In other cases we first compute an $\alpha \in \F$ of  
order larger than $(n-1)(k-1)$ and then return $M_k=[Z_1^k,\ldots,Z_m^k]$.  
We then use Theorem~\ref{main:thm: alpha folded wronskian} and properties of  Wronskian determinant to prove the correctness of our algorithm. 
Observe that when $M$ is a representation of a linear matroid then $M_k$ is a representation of it's $k$-truncation.
Further, each entry of $M_k$ is a polynomial of degree at most $n-1$ in $\F[X]$. 
Thus, testing whether a set of columns of size at most $k$ is independent,
 reduces to testing whether a determinant polynomial of degree at most $(n-1)k$ is identically zero or not.
This is easily done by evaluating the determinant at $(n-1)k + 1$ points in $\F$ and testing if it is zero at all those points.



\paragraph{Applications.} 
Matroid theory has found many algorithmic applications, starting from the characterization of greedy 
algorithms, to designing fixed parameter tractable (FPT) algorithms and kernelization algorithms. 
Recently the notion of {\em representative families} over linear matroids was used in designing 
fast FPT, as well as kernelization algorithm for several problems
\cite{FominLPS14, FominLS14,KratschW12,kratsch2012compression,Marx09, GoyalMP13,ShachnaiZ14}.   
Let us introduce this notion more formally.
Let \mat{} be a matroid and let  ${\cal S}=\{S_1, \dots, S_t\}$ be a family of subsets of $E$ of size  $p$. 
A subfamily $\widehat{\cal{S}}\subseteq \cal S$ is {\em $q$-representative} for $\cal S$ if  for every set $Y\subseteq  E$ of 
size at most $q$, if there is a set $X \in \cal S$ disjoint from $Y$ with $X\cup Y \in \I$, then there is a set 
$\widehat{X} \in \widehat{\cal S}$ disjoint from $Y$  and $\widehat{X} \cup  Y \in \I$. In  other words, if a set $Y$ of size 
at most $q$ can be extended  to an independent set of size $|Y|+p$ by adding a subset from $\cal S$, then it also can be extended 
to an independent set of size $|Y|+p$ by adding a subset from $ \widehat{\cal S}$ as well. The Two-Families Theorem of 
Bollob{\'a}s \cite{Bollobas65} for extremal set systems and its generalization  to subspaces of a vector space of 
Lov{\'a}sz \cite{Lovasz77} (see also \cite{Frankl82}) imply that every family of sets of size $p$ has a $q$-representative family 
with at most $\binom{p+q}{p}$ sets. Recently, Fomin et. al.~\cite{FominLS14} gave an efficient randomized algorithm to compute a representative family of size $\binom{p+q}{p}$  in a linear matroid of rank $n>p+q$.  This algorithm starts by computing 
a randomized $p+q$-truncation of the given linear matroid and then computes a $q$-representative family over the truncated matroid deterministically. 
Therefore one of our  motivations to study the  $k$-truncation  problem was to 
find an efficient deterministic computation of a representative family in a linear matroid. 
Formally, we have the following theorem.
\begin{theorem}
\label{main:thm:repsetlovaszrandomized}
Let \mat{}   be a linear matroid  of rank $n$ and let $ \cS = \{S_1,\ldots, S_t\}$ be a $p$-family of independent sets. 
Let $A$ be a $n\times |E|$ matrix representing $M$ over a field $ \mathbb{F}$, where $\F=\F_{p^\ell}$ 
or  $\F$ is ${\mathbb Q}$. Then there  are deterministic algorithms  computing  
\rep{S}{q}  as follows. 
\begin{enumerate}
\setlength{\itemsep}{-2pt}
\item A family $\widehat{\cal S}$ of size \bnoml{p+q}{p} in  $\cO\left({p+q \choose p}^2 t p^3n^2 + t {p+q \choose q} ^{\omega} np\right)+(n+|E|)^{\cO(1)}$, operations over $ \mathbb{F}$. 
\item A family   $\widehat{\cal S}$ of size  $ np {p+q \choose p}$ in  $\cO\left({p+q \choose p} t p^3n^2 
+ t {p+q \choose q}^{\omega-1} (pn)^{\omega-1} \right)+(n+|E|)^{\cO(1)}$ operations over $ \mathbb{F}$.  
\end{enumerate}
\end{theorem}

As a corollary of the above theorem, we obtain a deterministic FPT algorithm  {\sc $\ell$-Matroid Intersection},
derandomizing the algorithm of Marx~\cite{Marx09}. 
Using our results one can compute, in deterministic polynomial time, the $k$-truncation of graphic and co-graphic matroids,
which has important applications in graph algorithms.


\section{Preliminaries}
%
%
%

%
%
%
In this section we give various definitions and notions which we make use of in the paper.  We use the following notations: 
 $[n]=\{1,\ldots,n\}$ and ${[n] \choose i}=\{X~|~X\subseteq [n],~|X|=i\}$.

\subsection{Fields and Polynomials}
In this section we review some definitions and properties of fields. We refer to any graduate text on algebra for more details.  
The cardinality or the size of a field is called its {\em order}.  
It is well known that rational numbers ${\mathbb Q}$ and real numbers ${\mathbb R}$
are fields of infinite order. For every prime number $p$ and a positive integer $l$, there exists a finite field  of  
order  $p^l$. 
For a prime number $p$, the set $\{0,1, \ldots, p-1\}$ with addition and multiplication modulo $p$ forms 
a field, which we denote by $\F_p$.
Such fields are known as {\em prime fields}. 
Let $\F$ be a finite field and $\F[X]$ be the ring of polynomials in $X$ over $\F$.
For the ring $\F[X]$ we can define a field $\F(X)$ which is called the
{\em field of fractions} of $\F[X]$ as follows. The elements of $\F(X)$
are of the form $P(X) / Q(X)$ where $P(X), Q(X) \in \F[X]$ and $Q(X)$ is not a zero polynomial. The addition and multiplication
operations from $\F[X]$ are extended to $\F(X)$ in the usual way.  The degree of a polynomial $P(X)\in \F(X)$ is the highest 
exponent on indeterminate $X$ with a nonzero coefficient in $P(X)$. 
We will use $\F[X]^{<n}$ to denote the set 
the polynomials in $\F[X]$ of degree $<n$. 

For a field $\F$, we use $+_{\F}$ and $\times_{\F}$ to denote the addition and multiplication 
operations. (Often we write $a + b$ and $ab$ when the context is clear.)
The {\em characteristic} of a field, denoted by  $\ch(\F)$, is defined as least integer $m$  
such that $\sum_{i=1}^{m} 1 = 0$.
For fields such as ${\mathbb R}$ where there is no such $m$, the characteristic is defined to be $0$. For a finite field $\F = \F_{p^l}$, let $\F^* = \F \setminus \{ 0 \}$.
This is called the multiplicative group of $\F$
which is a cyclic group and has a generator $\alpha \in \F^*$. 
Every element of $\F^*$ can be written as $\alpha^i$ for some number $i$.
The element $\alpha$ is called a {\em primitive element} of the field $\F$. We say that an element $\beta \in \F$ has {\em order $r$}, if $r$ is the least integer such that $\beta^r=1$.

A polynomial $P(X) \in \F[X]$ is called {\em irreducible} if it cannot expressed 
as a product of two other non-trivial polynomials in $\F[X]$. 
Let $P(X)$ be an irreducible polynomial in $\F[X]$, of degree $l$.
Then $\K = \frac{\F[X]}{P(X)} = \F[X](\mathrm{ mod}~ P(X))$ is also a field.
It is of order $|\F|^l$ and characteristic of $\K$ is equal to the characteristic of $\F$.
We call $\K$ a {\em field extension} of $\F$ of degree $l$.
All finite fields are obtained as extensions of prime fields, 
and for any prime $p$ and positive integer $l$ there is exactly one finite field of order $p^l$
up to isomorphism.



\subsection{Vectors and Matrices}
A vector $v$ over a field $\F$ is an array of values from $\F$.
A collection of vectors $\{v_1, v_2, \ldots, v_k\}$ are said to be linearly dependent
if there exist values $a_1, a_2, \ldots, a_k$, not all zeros,  from $\F$ 
such that $\sum_{i=1}^{k} a_i v_i = 0$. Otherwise these vectors are called linearly
independent. 

For a matrix $A=(a_{ij})$ over a field $\F$, the row set and the 
 column set are denoted by $\mathbf{R}(A)$ and $\mathbf{C}(A)$ respectively. For $I \subseteq \mathbf{R}(A)$ and $J\subseteq \mathbf{C}(A)$, 
 $A[I,J]=\big(a_{ij}~|~i \in I,~j\in J\big)$ means the submatrix (or minor) of $A$ with the row set $I$ and the column set $J$.  
The matrix is said to have dimension $n \times m$ if it has $n$ rows and $m$ columns. For a matrix $A$ (or a vector $v$) by $A^T$ 
(or $v^T$) we denoted its {\em transpose}. 
Note that each column of a matrix is a vector over the field $\F$. 
The rank of a matrix is the cardinality of the maximum sized collection of columns
which are linearly independent. Equivalently, the rank of a matrix is the maximum number
$k$ such that there is a $k \times k$ submatrix whose determinant is non-zero.
We can use the definition of rank of a matrix to certify the linear independence of
a set of vectors. A collection of $n$ vectors $v_1, v_2, \ldots, v_n$ are 
linearly independent if and only if the matrix $M$ formed by $v_1, v_2, \ldots, v_n$ 
as column vectors have rank $n$. Determinant of a $n\times n$ matrix $A$ is denoted by  $\mathrm{det}(A)$ and is defined as 
\[ \mathrm{det}(A)=\sum_{\sigma\in S_n} \mathrm{sgn}(\sigma)\prod_{i=1}^n A[i,\sigma(i)].\]
Here, $S_n$ is the set of all permutations of $\{1,\ldots,n\}$ and $\mathrm{sgn}(\sigma)$ denotes the signature of the permutation $\sigma$. 

Throughout the paper we use $\omega$ to denote the matrix multiplication exponent. The current best known bound on 
$\omega<2.373$~\cite{Williams12}. We use $e$ to denote the base of natural logarithm.





\subsection{Matroids}
In the next few subsections we give definitions related to matroids. For a broader overview on matroids we refer to~\cite{oxley2006matroid}. 
\begin{definition}
A pair \mat, where $E$ is a ground set and $\cal I$ is a family of subsets (called independent sets) of $E$, is a {\em matroid} if it satisfies the following conditions:
 \begin{enumerate}
 \item[\rm (I1)]  $\phi \in \cal I$. 
 \item[\rm (I2)]  If $A' \subseteq A $ and $A\in \cal I$ then $A' \in  \cal I$. 
 \item[\rm (I3)] If $A, B  \in \cal I$  and $ |A| < |B| $, then there is $ e \in  (B \setminus A) $  such that $A\cup\{e\} \in \cal I$.
 \end{enumerate}
\end{definition}
The axiom (I2) is also called the hereditary property and a pair $(E,\cal I)$  satisfying  only (I2) is called hereditary family.  
An inclusion wise maximal set of $\cal I$ is called a {\em basis} of the matroid. Using axiom (I3) it is easy to show that all the bases of a matroid  
have the same size. This size is called the {\em rank} of the matroid $M$, and is denoted by \rank{M}. 
 \subsubsection{Linear Matroids and Representable Matroids} 
 Let $A$ be a matrix over an arbitrary field $\mathbb F$ and let $E$ be the set of columns of $A$. For $A$, we define  matroid 
 \mat{} as follows. A set $X \subseteq E$ is independent (that is $X\in \cal I$) if the corresponding columns are linearly independent over $\mathbb F$.  
%
%
%
The matroids that can be defined by such a construction are called {\em linear matroids}, and if a matroid can be defined by a matrix $A$ over a 
field $\mathbb F$, then we say that the matroid is representable over $\mathbb F$. That is, a matroid \mat{} of rank $d$ is representable over a field 
$\mathbb F$ if there exist vectors in $\mathbb{F}^d$  corresponding to the elements such that  linearly independent sets of vectors 
 correspond to independent sets of the matroid. 
      A matroid \mat{}  is called {\em representable} or {\em linear} if it is representable over some field $\mathbb F$ 
   and the corresponding matrix is called the {\em representation matrix} of $M$.

\subsubsection{Truncation of a Matroid.} The {\em $t$-truncation} of a matroid \mat{}  is a matroid $M'=(E,{\cal I}')$ such that $S\subseteq  E$ is independent in $M'$  if and only if $|S| \leq t$ and $S$ is independent in $M$ (that is $S\in \cal I$).  

\begin{proposition}[{\cite[Proposition 3.7]{Marx09}}]
\label{prop:truncationrep}
Given a matroid $M$ with a representation $A$ over a finite field $\mathbb F$ and an integer $t$,  a representation of the $t$-truncation $M'$ 
can be found in randomized polynomial time.
\end{proposition}

\subsection{Derivatives}
Recall the definition of the formal derivative $\frac{d}{dx}$ of a function over $\mathbb R$.
We denote the $k$-th formal derivative of a function $f$ by $f^{(k)}$. 
We can extend this notion to finite fields.
%
%
%
%
%
%
Let $\F$ be a finite field and let $\F[X]$ be the ring of polynomials in $X$ over $\F$.
Let $P \in \F[X]$ be a polynomial of degree $n-1$, i.e. $ P = \sum_{i=0}^{n-1} a_i X^i $ where $a_i \in \F$.
Then we define the {\em formal derivative} of as $ P' = \sum_{i=1}^{n-1} ia_i X^{i-1} $.
We can extend this definition to the $k$-th formal derivative of $P$ as
$P^{(k)} = (P^{(k-1)})'$. Note that higher derivtives are defined iteratively 
using lower derivatives, thus they are also called iterated formal derivatives.

Formal derivatives continue to carry many of their properties in $\mathbb R$ 
to finite fields $\F$. However not all properties carry through. 
For example, in $\F_3[X]$ the polynomial $P(X) = X^3$ has all derivatives $0$.
%
%
To remedy such problems, we require the notion of {\em Hasse Derivatives}.
For a polynomial $P(X) \in \F[X]$, the $i$-th Hasse derivative $D^i(P)$ is defined as the 
coefficient of $Z^i$ in $P(X + Z)$. 
    $$ P(X + Z) = \sum_{i=0}^{\infty} D^i(P(X)) Z^i $$

We note some important properties of Hasse derivatives and how they relate to formal derivatives.
We refer to \cite{dvir2009extensions} and \cite{goldschmidt2003algebraic} for details.

\begin{lemma} \label{thm: hasse derivative properties}
Let $\F$ be a finite field of characteristic $p$, $P, Q \in \F[x]$.
Then the following holds :
\begin{enumerate}

    \item $D^k$ is a linear map from $\F[X]$ to $\F[X]$.

    \item Let $k$ be some number and let $k!$ be non-zero in $\F$ (i.e. $k! \neq 0 \mod p$).
          Then $k! . D^k(P) = P^{(k)}$. In particular $D^0(P) = f$ and $D^1(P) = P^{(1)}$.

\end{enumerate}
\end{lemma}

%
%
%
%
Observe that the second statement in the above lemma shows that the value of $k$-th hasse derivatives 
and $k$-th formal derivatives differ only by a multiplicative value whenever $k!$ is non-zero in $\F$.
In particular when $k < p$, (the characteristic of the field $\F$) then $k!$ is always non-zero in $\F$.
In our setting this is always the case.

\section{Matroid Truncation}
In this section we give our one of the  main results.  We start with an introduction to our tools and then we give two 
results that give rank $k$-truncation of the given matrix $M$.

\subsection{Tools and Techniques}
In this subsection we collect some known results, definitions and derive some new connections among them that 
will be central to our results.  We also prove one of the main technical contribution of the paper in this subsection (Theorem~\ref{thm: alpha folded wronskian}). 

\subsubsection{Polynomials and Vectors}
Let $\F$ be a field. 
The set of polynomials $P_1(X), P_2(X), \ldots, P_k(X)$ in $\F[X]$ are said to be {\em linearly independent} over $\F$  if there 
doesn't exist $a_1, a_2, \ldots, a_k \in \F$, not all zeros such that 
$\sum_{i=1}^{k} a_i P_i(X) \equiv 0$. Otherwise they are said to be linearly dependent. 

\begin{define}
Let $P(X)$ be a polynomial of degree at most $n-1$ in $\F[X]$.
We define the vector $v$ corresponding to the polynomial $P(X)$ as follows: 
$v[j] = c_{j} \text{ where }  P(X) = \sum\limits_{j=1}^{n} c_{j} x^{j-1}.$ 
Similarly given a vector $v$ of length $n$ over $\F$,
we define the polynomial $P(X)$ in $\F[X]$ corresponding to the vector $v$ as follows: 
$ P(X) = \sum\limits_{j=1}^{n} v[j] x^{j-1} .$
\end{define}

%

The next lemma will be key to several proofs later.  
\begin{lemma}
\label{lemma: vectors and polynomials} 
Let $v_1, \ldots, v_k$ be vectors of length $n$ over $\F$ and
let $P_1(X),  \ldots, P_k(X)$ be the corresponding polynomials respectively.
Then $P_1(X),  \ldots, P_k(X)$ are linearly independent over $\F$ if and only if
$v_1, v_2, \ldots, v_k$ are linearly independent over $\F$.
\end{lemma}
\begin{proof}  For $i \in \{1 \ldots k\}$, let $v_i = (c_{i1}, \ldots, c_{in})$ and let 
$P_i(X) = \sum_{j=1}^{n} c_{ij}x^{j-1}$ be the polynomial corresponding to $v_i$. 

We first prove the forward direction of the proof. For a contradiction, assume that 
$v_1, \dots, v_k$ are linearly dependent.
Then there exists $a_1, \ldots, a_k \in \F$, not all zeros, such that
$\sum_{i=1}^{k} a_i v_i = 0$.
This implies that for each $j \in \{1, \ldots n\}$, $\sum_{i=1}^{k} a_i v_i[j] = 0$.
Since $v_i[j] = c_{ij}$, we have $\sum_{i=1}^{k} a_i c_{ij} = 0$,
which implies that $\sum_{i=1}^{k} a_i c_{ij} x^{j-1} = 0$.
Summing over all these expressions we get $\sum_{i=1}^{k} a_i P_i(X) \equiv 0$, a contradiction. This completes the proof in the 
forward direction. 

Next we prove the reverse direction of the lemma. To the contrary assume that 
$P_1(X),  \ldots, P_k(X)$ are linearly dependent.
Then there exists $a_1, \ldots, a_k \in \F$, not all zeros, such that
$\sum_{i=1}^{k} a_i P_i(X) \equiv 0$. 
This implies that for each  $j \in \{1, \ldots, n\}$,  the coefficients of $x^{j-1}$ satisfy
$\sum_{i=1}^{k} a_i c_{ij} = 0$.
Since $c_{ij}$ is the $j$-th entry of the vector $v_i$ for all $i$ and $j$, we have
$\sum_{i=1}^{k} a_i v_i = 0$.
Thus the vectors $v_1,  \ldots, v_k$ are linearly dependent, a contradiction to the given assumption. 
This completes this proof. 
\end{proof}

We will use this claim to view the column vectors of a matrix $M$ over a field $\F$ 
as elements in the ring $\F[X]$ and  in the field of fractions $\F(X)$. 
We shall then use properties of polynomials to deduce properties of these
column vectors and vice versa.

\subsubsection{Wronskian}
Let  $\F$ be a field with characteristic at least $n$. Consider a collection of polynomials $P_1(X),  \ldots, P_k(X)$ from $\F[X]$ of degree at most $n-1$. 
We define the following matrix, called the {\em Wronskian}, of $P_1(X),  \ldots, P_k(X)$ 
as follows. 

$$ 
    W(P_1(X),  \ldots, P_k(X)) =   W(P_1,  \ldots, P_k) = \begin{pmatrix}
                                P_1(X)         & P_2(X)         & \ldots & P_k(X)          \\
                                P_1^{(1)}(X)   & P_2^{(1)}(X)   & \ldots & P_k^{(1)}(X)    \\
                                \vdots      & \vdots      & \ddots & \vdots       \\
                                P_1^{(k-1)}(X) & P_2^{(k-1)}(X) & \ldots & P_k^{(k-1)}(X)    
                                \end{pmatrix}_{k \times k}
$$

Note that, the determinant of the above matrix actually yields a polynomial.  For our purpose we will need the following well known result.
\begin{theorem} [\cite{bostan2010wronskians, garcia1987wronskians, muir1882treatise}]
\label{thm: classic wronskian}
Let  $\F$ be a field and $P_1(X),  \ldots, P_k(X)$ be a set of polynomials from $\F[X]^{<n}$
and let $\ch(\F)>n$. 
Then  $P_1(X),  \ldots, P_k(X)$  are linearly independent over $\F$ if and only if   the  Wronskian determinant
$\mathrm{det}(W(P_1(X),  \ldots, P_k(X)))\not\equiv 0$ in $\F[X]$. 
\end{theorem}

The notion of Wronskian dates back to 1812~\cite{muir1882treatise}. We refer to~\cite{bostan2010wronskians, garcia1987wronskians}  for some recent variations and proofs. The switch between usual derivatives and Hasse derivatives multiplies the Wronskian determinant by a constant, which is non-zero as long as $n< \ch(\F)$, and thus this criterion works with both notions. 
Observe that the  Wronskian determinant is a polynomial of degree at most $nk$ in $\F[X]$.
Thus to test if such a polynomial (of degree $d$) is identically zero, 
we only need to evaluate it at $d+1$ arbitrary points of the field $\F$, 
and check if it is zero at all those points. 

\subsubsection{Folded Wronskian}
The above definition of Wronskian requires us to compute derivatives of 
degree $(n-1)$ polynomials. As noted earlier, they are well defined only if the underlying field has characteristic greater 
than or equal to $n$. However, we  might have matrix over fields of small characteristic.  For these kind of fields, we have the 
notion of {\em Folded Wronskian} which  was recently introduced by Guruswami and Kopparty in the context of subspace design~\cite{GuruswamiK13}.

Consider a collection of polynomials $P_1(X),  \ldots, P_k(X)$  from ${\mathbb F} [X]$ of degree at most $(n-1)$.
Further, let $\F$ be of order at least $n+1$,  and  $\alpha$ be an element of $\F^*$.  
We define the the {\em $\alpha$-folded Wronskian}, of $P_1(X),  \ldots, P_k(X)$ 
as follows. 

$$ 
    W_\alpha(P_1(X),  \ldots, P_k(X)) =    W_\alpha (P_1,  \ldots, P_k) = 
    \begin{pmatrix}
        P_1(X)              & P_2(X)                & \ldots & P_k(X)                \\
        P_1(\alpha X)       & P_2(\alpha X)        & \ldots & P_k(\alpha X)         \\
        \vdots              & \vdots                & \ddots & \vdots                \\
        P_1(\alpha^{k-1} X) & P_2(\alpha^{k-1} X)   & \ldots & P_k(\alpha^{k-1} X)   \\    
    \end{pmatrix}_{k \times k}
$$
As before, the determinant of the above matrix is a polynomial of degree at most $nk$ in $\F[X]$. 
It is important to note that Guruswami and Kopparty  used the above notion of  $\alpha$-folded Wronskian only for those 
$\alpha$ that are {\em primitive element} of $\F$.  In particualr, they proved 
the following result about $\alpha$-folded Wronskians~\cite[Lemma~12]{GuruswamiK13}. 
\begin{theorem}[\cite{GuruswamiK13}]
\label{thm: folded wronskian}
Let $\F$ be a field of order $>n$, $\alpha$ be a primitive element of $\F$ and let 
$P_1(X),  \ldots, P_k(X)$ be a set of polynomials from $\F[X]^{<n}$. 
Then $P_1(X),  \ldots, P_k(X)$ are linearly independent over $\F$ if and only if the $\alpha$-folded Wronskian 
determinant 
$\mathrm{det}(W_\alpha(P_1,  \ldots, P_k))\not\equiv 0$ in $\F[X]$.  
\end{theorem}




Theorem~\ref{thm: folded wronskian} 
requires a primitive element $\alpha$ of the underlying field $\F$. However, finding a primitive element in a finite field is a non-trivial problem
and currently there are no deterministic polynomial time algorithm known for this problem. To  overcome this difficulty, we prove a generalization of Theorem~\ref{thm: folded wronskian}. 
We relax the requirement that $\alpha$ must be a primitive element of 
the field $\F$ and only require that $\alpha$ be an element of {\em sufficiently large} order. Finding an element of large order 
is slightly easier task than finding a primitive element of a  finite field $\F$. We will see that for our purpose this will be sufficient.  
This result  will be our main technical tool. The proof for next theorem is very different than the one for Theorem~\ref{thm: folded wronskian}. Theorem~\ref{thm: folded wronskian}  crucially uses the fact that $\alpha$ is a primitive element.  This is used to define an irreducible polynomial which is key to the arguments used in~\cite[Lemma~12]{GuruswamiK13}. We do not see ways to make the arguments used for the proof of  Theorem~\ref{thm: folded wronskian}  to go through for our case.  Our theorem is as follows. 

%
\begin{theorem}[Theorem ~\ref{main:thm: alpha folded wronskian}, restated]
\label{thm: alpha folded wronskian}
Let $\F$ be a field, $\alpha$ be an element  of $\F$ of order $>(n-1)(k-1)$ and let 
$P_1(X),  \ldots, P_k(X)$ be a set of polynomials from $\F[X]^{<n}$. 
Then $P_1(X),  \ldots, P_k(X)$ are linearly independent over $\F$ if and only if the $\alpha$-folded Wronskian 
determinant 
$\mathrm{det}(W_\alpha(P_1, \ldots, P_k))\not\equiv 0$ in $\F[X]$.  
\end{theorem}
%
In what follows we build towards the proof of Theorem~\ref{thm: alpha folded wronskian}.  
For the sake of brevity, we use $W_\alpha$ to denote the 
matrix  $W_\alpha (P_1,  \ldots, P_k)$. 
We use the notation $Z_1, Z_2, \ldots, Z_k$ 
to denote the columns 
of $W_\alpha$.  That is, $Z_i=(P_i(X),P_i(\alpha X),\ldots,P_i(\alpha^{k-1}X))^T$. 
Observe that $W_\alpha$ is a matrix over the field $\F(X)$, with entries from the ring $\F[X]$.
When we talk about linear independence of $\{Z_i\}_{i=1}^k$, the underlying field is $\F(X)$.
We recall the following well known lemma about non-zero determinant of a square matrix 
and the linear independence of it's columns.

\begin{lemma} \label{lemma: det and lin ind}
    Let $M$ be a $n\times n$ over a field $\F$. 
    Then $\det(M) \not = 0$ if and only if the columns of $M$ are linearly independent over $\F$.
\end{lemma}
The next lemma will  prove the reverse direction of  Theorem~\ref{thm: alpha folded wronskian}. 

\begin{lemma} 
\label{lem:foldedwronforward}
 If $P_1(X), \ldots, P_k(X)$ are linearly dependent over $\F$, then $\det(W_\alpha) \equiv 0$.
\end{lemma}
\begin{proof}
Since  $P_1(X),\ldots, P_k(X)$ are linearly dependent over $\F$,  there exist $\lambda_1,  \ldots, \lambda_k \in \F$ (not all equal to zero)  such that  $\sum_{i=1}^{k} \lambda_i P_i(X) = 0$. Therefore, for all $ d \in \{0, 1, \ldots, k-1\}$ we have that 
$\sum_{i=1}^{k} \lambda_i P_i(\alpha^d X) = 0$. This implies that $\sum_{i=1}^{k} \lambda_i Z_i = 0$. That is, the columns of 
$W_\alpha$ are linearly dependent over $\F\subseteq \F(X)$. Therefore by Lemma~\ref{lemma: det and lin ind}  
the polynomial $\det(W_\alpha)$ is identically zero. That is,  $\det(W_\alpha) \equiv 0$. This completes the proof. 
\end{proof}

The next lemma will be used in the forward direction of the proof. 
\begin{lemma} \label{lemma: poly multiple}
Let $A(X)$ and $B(X)$ be non zero polynomials in $\F[X]$ of degree at most $l$. Let $\beta \in \F$ be an element of order $>l$. 
If $A(X)B(\beta X) - A(\beta X)B(X) \equiv 0 $ then $A(X) = \lambda B(X)$ where $0\neq\lambda \in \F$.
\end{lemma}
\begin{proof}
Let $A(X) = \sum_{i=0}^l a_i X^i$, and $B(X) = \sum_{j=0}^l b_j X^j$ where $a_i, b_j \in \F$. 

\paragraph{Case 1:}  We first assume that 
 $b_0\neq 0$.  Later we will show how all other cases reduce to this.  
%
Let $$S_{A,B}(X) = A(X)B(\beta X) - A(\beta X) B(X).$$ 
Since $S_{A,B}(X) \equiv 0$ we have that for all $t\in \{0, 1, \ldots, 2l\}$, the coefficients of $X^t$ in $S_{A,B}(X)$ is zero. 
Our proof is based on the following claim. 
\begin{claim} 
For all $i \in \{1, \ldots l\}$, either $\frac{a_i}{b_i} = \frac{a_0}{b_0}$, or $a_i = b_i = 0$.
\end{claim}
\begin{proof} 
We prove the claim using induction on $i$. For $i=1$, consider the coefficient of $X$ in $S_{A,B}(X)$. 
The coefficient of $X$ in $S_{A,B}(X)$ is $(\beta - 1)(a_0b_1 - a_1b_0)$. Since the order of $\beta$ is $>l$, 
$(\beta - 1)\neq 0$. This implies $(a_0b_1 - a_1b_0)=0$. So if $b_1=0$, then $a_1=0$ (because $b_0\neq 0$) 
and if $b_1\neq 0$, then $\frac{a_0}{b_0} = \frac{a_1}{b_1}$. 
Thus we assume that $i\geq 2$ and by induction hypothesis the claim holds for $j \in \{1, \ldots, i-1\}$.
Now, consider the coefficients of $X^i$ in $S_{A,B}(X)$. The coefficient of $X^i$ in $S_{A,B}(X)$ is 
$$\beta^i(a_0b_i - a_ib_0) + \beta^{i-1}(a_1b_{i-1} - a_{i-1}b_1) + \ldots + (a_ib_0 - a_0b_i).$$ 
By our assumption we know that 
$$\beta^i(a_0b_i - a_ib_0) + \beta^{i-1}(a_1b_{i-1} - a_{i-1}b_1) + \ldots + (a_ib_0 - a_0b_i) = 0.$$
Consider the term $\beta^j(a_{i-j}b_j - a_{j}b_{i-j})$ for $ 0 < j < i$. By induction hypothesis, one of the 
following statement is true. 
\begin{itemize}
 \item $a_j = b_j = 0$
 \item $a_{i-j} = b_{i-j} = 0$
 \item $\frac{a_{i-j}}{b_{i-j}} = \frac{a_0}{b_0} = \frac{a_j}{b_j}$
\end{itemize}
In all the three cases the term $\beta^j(a_{i-j}b_j - a_{j}b_{i-j})$ is zero. 
Therefore the coefficient of $X^i$ simplifies to, $(\beta^i - 1)(a_0b_i - a_ib_0)$ and we get $(\beta^i - 1)(a_0b_i - a_ib_0)=0$. 
Since the order of $\beta$ is $>l$, $(\beta^{i} - 1)\neq 0$. This implies $(a_0b_i - a_ib_0)=0$. 
So if $b_i=0$, then $a_i=0$ (because $b_0\neq 0$) and if $b_i\neq 0$, then $\frac{a_0}{b_0} = \frac{a_i}{b_i}$. This concludes the claim. 
\end{proof}
Let $\lambda = \frac{a_0}{b_0} \in \F$. Thus $a_i = \lambda b_i$. Therefore $A(X) = \lambda B(X)$. Since $A(X)\not\equiv 0$, $\lambda\neq 0$. 

\paragraph{Case 2:}  Suppose $b_0=0$ and $a_0\neq 0$. 
Then let $A_1(X)=B(X)$ and $B_1(X)=A(X)$. 
Since  $A(X)B(\beta X) - A(\beta X)B(X) \equiv 0 $, we have that 
\begin{eqnarray*}
 & & B_1(X)A_1(\beta X) - B_1(\beta X)A_1(X) \equiv 0 \\
&\implies & -(B_1(\beta X)A_1(X) - B_1(X)A_1(\beta X))\equiv 0. 
\end{eqnarray*} 
Thus $A_1(X)B_1(\beta X) - A_1(\beta X)B_1(X) \equiv 0 $.  So by applying  {\bf Case 1}  with $A_1(X)$ and $B_1(X)$, we get $A_1(X)=\lambda B_1(X)$ where $0\neq \lambda\in \F$.  This implies that $A(X)=\lambda^{-1}B(X)$ where 
$0\neq \lambda^{-1}\in \F$.

\paragraph{Case 3:}  Suppose $b_0=0$ and $a_0= 0$.   Let $r$ be the least integer such that either $a_r\neq 0$ or $b_r\neq 0$. 
Then let $A(X) = X^r A_2(X)$ and $B(X) = X^rB_2(X)$. 
Since $A(X)B(\beta X) - A(\beta X)B(X) \equiv 0 $, we have that $A_2(X)B_2(\beta X) - A_2(\beta X)B_2(X) \equiv 0 $. 
Note that the coefficient of $X^0$ is non zero in at least one of the polynomials $A_2(X)$ or $B_2(X)$. Furthermore,   $A_2(X),B_2(X)\not\equiv 0$.  Thus, if the coefficient of $B_2(X)$ is non-zero then we are  in {\bf Case 1} else we are in {\bf Case 2}.  
This completes the proof of the lemma.
\end{proof}

The next lemma will  be useful in showing the forward direction of Theorem~\ref{thm: alpha folded wronskian}. 

\begin{lemma} 
\label{lem:foldedwronreverse}
Let $P_1(X), \ldots, P_k(X)$ be a set of polynomials from   $\F[X]^{< n}$ and $\alpha$ be an element of order $>(n-1)(k-1)$.  
If $\det(W_\alpha) \equiv 0$, then $P_1(X),  \ldots, P_k(X)$ are linearly dependent over $\F$. 
\end{lemma}

\begin{proof}  
We prove the lemma using induction on $k$ -- the number of polynomials.  For $k$ = 1, $W_\alpha = [ P_1(X) ]$ and the lemma vacuously holds. Form now on we assume that $k\geq 2$ and that the lemma holds for all $t<k$.  Recall that $Z_i$ denotes the $i$-th column of the matrix $W_\alpha$. 
We first give the proof for the case when there is a subset of columns of $Z_1, \ldots, Z_k$, of size $<k$ that are   linearly dependent 
over $\F(X)$.  Let $\{i_1,\ldots,i_t\}\subset \{1,\ldots,k\}$ of size at most $k-1$, such that  
$Z_{i_1},\ldots,Z_{i_t}$ are linearly dependent. Thus, there exists $\delta_1(X),\ldots,\delta_t(X)\in \F(X)$, not all equal to zero 
polynomial in $\F(X)$,  such that $\sum_{j=1}^t \delta_i(X)Z_{i_j}\equiv 0 $. 
Let $Z_{i_1}',\ldots,Z_{i_t}'$  denote the 
restriction of  $Z_{i_1},\ldots,Z_{i_t}$  to first $t$ rows of $W_\alpha$. 
Since $\sum_{j=1}^t \delta_i(X)Z_{i_j}\equiv 0 $, this implies that $\sum_{j=1}^t \delta_i(X)Z_{i_j}'\equiv 0 $. 
Therefore $Z_{i_1}',\ldots,Z_{i_t}'$ are also linearly dependent over $\F(X)$. 
Consider $W_\alpha(P_{i_1}, \ldots, P_{i_t})$. This is a $t\times t$ matrix with columns $Z_{i_1}',\ldots,Z_{i_t}'$. Since $Z_{i_1}',\ldots,Z_{i_t}'$ are linearly dependent, by Lemma~\ref{lemma: det and lin ind}, 
$\det(W_\alpha(P_{i_1}, \ldots, P_{i_t})) \equiv 0$.   By induction hypothesis, this implies  that 
$P_{i_1}(X),  \ldots, P_{i_t}(X)$ are linearly dependent over ${\mathbb  F}$ and thus 
$P_1(X), \ldots, P_k(X)$ are linearly dependent over $\F$.  So from now on we  assume that 
for any subset $\{i_1,\ldots,i_t\}\subset \{1,\ldots,k\}$ of size at most $k-1$, 
$Z_{i_1},\ldots,Z_{i_t}$ are linearly independent and $\det(W_\alpha(P_{i_1}, \ldots, P_{i_t})) \not\equiv 0$. 

Next we prove the claim that if $\{Z_i\}_{i=1}^k$ are linearly dependent then we can choose the polynomials 
$\delta_i(X)$ in $\F(X)$ which satisfy certain desirable properties. 
\begin{claim} \label{lemma: a_i from det}
Let $\det(W_\alpha(P_1,\ldots,P_k)\equiv 0$. Then there exist non-zero polynomials $\delta_1(X),  \ldots, \delta_k(X)\in \F[X]$, of degree at most $(n-1)(k-1)$ 
such that $\sum_{i=1}^{k} \delta_i(X) Z_i = 0$.
\end{claim}
\begin{proof} 
For all $i\in\{1,\ldots,k\}$, define 
\begin{eqnarray*}
\delta_i(X) & = & (-1)^{1+i} \det(W_\alpha(P_{1}(\alpha X),\ldots,P_{i-1}(\alpha X), P_{i+1}(\alpha X), \ldots, P_{k}(\alpha X))) \\
                   &= &   (-1)^{1+i} \det(W_\alpha(P_{1},\ldots,P_k)[\{2,\ldots,k\},\{1,\ldots,k\}\setminus \{i\}]).
\end{eqnarray*}
Observe that $$\sum_{i=1}^{k} \delta_i(X) P_i(X) = \det(W_\alpha(P_{1},  \ldots, P_{k})) \equiv 0,$$
because by assumption  $\det(W_\alpha(P_1,\ldots,P_k))\equiv 0$.
Now consider the matrix $W_j$ obtained by replacing the first row of  $W_\alpha(P_1,\ldots,P_k)$ with $j^{th}$ row of $W_\alpha(P_1,\ldots,P_k)$. That is,  

    $$ 
    W_j= 
    \begin{pmatrix}
        P_1(\alpha^{j-1} X)              & P_2(\alpha^{j-1} X)                & \ldots & P_k(\alpha^{j-1} X)                \\
        P_1(\alpha X)       & P_2(\alpha X)        & \ldots & P_k(\alpha X)         \\
        P_1(\alpha^2 X)       & P_2(\alpha^2 X)        & \ldots & P_k(\alpha^2 X)         \\
        \vdots              & \vdots                & \ddots & \vdots                \\
        P_1(\alpha^{k-1} X) & p_2(\alpha^{k-1} X)   & \ldots & P_k(\alpha^{k-1} X)   \\    
    \end{pmatrix}_{k \times k}
    $$                                         
Note that for any $j\in\{2,\ldots,k\}$, $\sum_{i=1}^{k} \delta_i(X) P_i({\alpha}^{j-1}  X)= \det(W_j)\equiv 0$ (as $W_j$ has two identical rows). Hence, $\sum_{i=1}^{k} \delta_i(X) Z_i = 0$.
Since we are in the case where for any subset $\{i_1,\ldots,i_t\}\subset \{1,\ldots,k\}$ of size at most $k-1$, 
$Z_{i_1},\ldots,Z_{i_t}$ are linearly independent we have that  
$$\det(W_\alpha(P_{1},\ldots,P_{i-1}, P_{i+1}, \ldots, P_{k})) \not\equiv 0.$$
 This implies that 
$$\delta_i(X)=(-1)^{1+i}\det(W_\alpha(P_{1}(\alpha X),\ldots,P_{i-1}(\alpha X), P_{i+1}(\alpha X), \ldots, P_{k}(\alpha X))) \not\equiv 0$$
and the degree of $\delta_i(X)$ is at most $(n-1)(k-1)$ for all $i\in\{1,\ldots,k\}$. This completes the  proof of the claim.  
    \end{proof}
From now on we work with the collection $\{\delta_i(X)\}_{i=1}^k$ provided by Claim~\ref{lemma: a_i from det}.
We have the following expression.
$$ \sum_{i=1}^{k} \delta _i(X) Z_i =\sum_{i=1}^{k} \delta_i(X) \begin{bmatrix} P_i(X) \\ P_i(\alpha X) \\ \ldots \\ P_i(\alpha^{k-1} X) \end{bmatrix}  = 0 $$
This implies that  for each $j\in\{0,\ldots,k-1\}$, we have
\begin{alignat}{1}
            \sum_{i=1}^{k}   \delta_i(X) P_i(\alpha^j X)  
                            & =  0 . \label{eqn:1} \\
        \intertext{By rearranging Equation~\ref{eqn:1} and absorbing the negative sign into $\delta_k(X)$ we get }
            \sum_{i=1}^{k-1} \delta_i(X) P_i(\alpha^j X) 
                            & =  \delta_k(X) P_k(\alpha^j X) . \label{eqn:2} \\
        \intertext{Substitute $\alpha X$ for $X$ in Equation~\ref{eqn:2}  for all $j\in \{0,\ldots,k-2\}$ we get} 
            \sum_{i=1}^{k-1} \delta_i(\alpha X) P_i(\alpha^{j+1} X) 
                            & =  \delta_k(\alpha X) P_k(\alpha^{j+1} X) . \label{eqn:3} \\
        \intertext{Substitute the value of $P_k(\alpha^{j+1} X)$ from Equation~\ref{eqn:2} into Equation~\ref{eqn:3}, we get that for all $j\in \{0,1,\ldots,k-2\}$} 
             \delta_k(X) \sum_{i=1}^{k-1} \delta_i(\alpha X) P_i(\alpha^{j+1} X)
                            & =  \delta_k(\alpha X) \sum_{i=1}^{k-1} \delta_i(X) P_i(\alpha^{j+1} X). \\ 
        \intertext{Thus for each $j\in \{0,1,\ldots,k-2\}$ we have}
            \sum_{i=1}^{k-1}\Big\{ \delta_k(X)\delta_i(\alpha X) - \delta_k(\alpha X) \delta_i(X) \Big\} P_i(\alpha^{j+1} X) 
                            & = 0. \label{eqn:5}\\
        \intertext{Substitute $\beta X$ for $X$ in Equation~\ref{eqn:5}, where $\beta = \alpha^{-1}$. Then  for all $j\in \{0,1,\ldots,k-2\}$ } 
            \sum_{i=1}^{k-1} \Big\{ \delta_k(\beta X)\delta_i(X) - \delta_k(X)\delta_i(\beta X) \Big\} P_i(\alpha^j X) 
                            & = 0. \label{eqn:6}\\
        \intertext{Let $Z_i'$ be the column vector corresponding to $Z_i$ in the matrix $W_\alpha(P_1,\ldots,P_k)$
                   restricted to the first $k-1$ rows. Then from Equation~\ref{eqn:6} we get that  
                   for all $j\in \{0,1,\ldots,k-2\}$  }
            \sum_{i=1}^{k-1} \Big\{ \delta_k(\beta X)\delta_i(X) - \delta_k(X)\delta_i(\beta X)\Big \}
                            \begin{bmatrix} P_i(X) \\ P_i(\alpha X) \\ \ldots \\ P_i(\alpha^{k-2} X) \end{bmatrix}
                            & = 0. \\
         \implies   \sum_{i=1}^{k-1} \Big\{ \delta_k(\beta X)\delta_i(X) - \delta_k(X)\delta_i(\beta X)\Big \} Z_i' & = 0. 
    \end{alignat}

    Consider the $(k-1) \times (k-1)$ matrix $[Z_1', Z_2', \ldots, Z_{k-1}'] = W_\alpha(P_1(X), P_2(X), \ldots, P_{k-1}(X))$. By the case of proof we are currently dealing, we have that  $W_\alpha(P_1(X), P_2(X), \ldots, P_{k-1}(X))$ has a non-zero determinant. In other words, the vectors $Z_1', \ldots, Z_{k-1}'$ are linearly independent over $\F(X)$. 
This implies that for all $ i\in \{1,\ldots,k-1\}$,
 $$ \delta_i(X)  \delta_k(\beta X)- \delta_i(\beta X) \delta_k(X) = 0$$
Observe that  $\delta_i(X)$,  $i\in \{1,\ldots,k\}$, are non-zero polynomials in $\F[X]$ of degree at most 
$(n-1)(k-1)$.  Furthermore, the order of $\beta$ is $>(n-1)(k-1)$, and thus by applying Lemma~\ref{lemma: poly multiple} we get 
that $\delta_i(X) = \lambda_i \delta_k(X)$ for all 
$i\in\{1,\ldots,k-1\}$ where $0\neq \lambda_i\in{\mathbb F}$. 
Now by simplifying we get the following expressions
\begin{eqnarray*}
         & &    \sum_{i=1}^{k}  \delta_i(X) Z_i  =  0  \\
         &\iff &   \sum_{i=1}^{k} \lambda_i \delta_k(X) Z_i  =0   \qquad\quad\mbox{ (because $\delta_i(X) = \lambda_i \delta_k(X)$)}\\
            &\iff&   \sum_{i=1}^{k} \lambda_i Z_i  =  0  \qquad\quad\mbox{ (because $\delta_k(X)\not\equiv 0$)} \\
            &\iff &   \sum_{i=1}^k \lambda_i P_i(X) =0 . 
\end{eqnarray*} 
Hence $P_1(X), \ldots P_k(X)$ are linearly dependent over $\F$. This completes the proof of the lemma. 
\end{proof}

Combining Lemmata~\ref{lem:foldedwronforward} and \ref{lem:foldedwronreverse}, we get the proof of Theorem~\ref{thm: alpha folded wronskian}.  We can get the following corrollary from Theorems~\ref{thm: folded wronskian} and \ref{thm: alpha folded wronskian}. 
\begin{corollary}
Let $\F$ be a field of size $>n$, $\alpha$ be either a primitive element or an element  of $\F$ of order $>(n-1)(k-1)$ and let 
$P_1(X),  \ldots, P_k(X)$ be a set of polynomials from $\F[X]^{<n}$. 
Then $P_1(X),  \ldots, P_k(X)$ are linearly independent over $\F$ if and only if the $\alpha$-folded Wronskian 
determinant 
$\mathrm{det}(W_\alpha(P_1, \ldots, P_k))\not\equiv 0$ in $\F[X]$.  
\end{corollary}




\subsubsection{Finding irreducible polynomials and elements of large order}
Whenever we will need to use folded Wronskians, we will also need to get hold of a primitive element  or an element of large 
order  of an appropriate field. 
We start by reviewing some known algorithms for finding irreducible polynomials over finite fields.
Note that for a finite field of order $p^l$, the field operations can be done in time $(l \log p)^{\cO(1)}$.
And for an  infinite field, the field operations will require $(\log  N)^{\cO(1)}$ where 
$N$ is the size of the largest value handled by the algorithm. 
Typically we will provide an upper  bound on $N$ when the need arises.
A result by Shoup \cite[Theorem 4.1]{Shoup88}) allows us to find an irreducible polynomial
of degree $r$ over $\F_{p^l}$ in time polynomial in $p, l$ and $d$.   Adleman and Lenstra \cite[Theorem 2]{AdlemanLenstra86}  gave an algorithm that allows us to compute  an irreducible polynomial of degree at least $r$ over a prime field $\F_p$
in time polynomial in $\log p$ and $r$.
\begin{lemma} [\cite{AdlemanLenstra86,Shoup88}]{\rm (Finding Irreducible Polynomials)}
\label{lemma: irred shoup}
\begin{enumerate}
\item  There is an algorithm that given prime $p$ and $r$, it can compute
an irreducible polynomial $f(X) \in \F_p[X]$ such that $r \leq \mathrm{deg}(f) \leq cr \log p$
in $(c r (\log p)^2)^c$ time, where $c$ is a constant.
    \item For $q = p^l$ and an integer $r$, we can compute an irreducible polynomial of $\F_q[X]$ of degree $r$
    in $\cO( \sqrt{p} (\log p)^3 r^3 (\log r)^{\cO(1)} + (\log p)^2 r^4 (\log r)^{\cO(1)} + (\log p) r^4 (\log r)^{\cO(1)} l^2 (\log l)^{\cO(1)} )$ time.
\end{enumerate}
\end{lemma}

%

%

Next we consider a few algorithms for finding primitive elements in finite fields. For fields of  large order but small characteristic, we have the following lemma,
which follows from the results of Shparlinski \cite{Shparlinski90} and
also from the results of Shoup \cite{Shoup90}.
\begin{lemma} [\cite{Shparlinski90,Shoup90}]
\label{lemma: primitive set}
Let $\F = \F_{p^l}$  be a finite field. Then we can compute a set $S \subset \F$, containing a primitive element,  
of size $\poly(p,l)$ in time $\poly(p,l)$.
\end{lemma}

We use Lemma~\ref{lemma: primitive set} to get the following result  
that allows us to find elements of sufficiently large order in a finite field or a primitive element in a field of small size.
\begin{lemma} \label{lemma: large order element}
Let $\F = \F_{p^l}$ be a finite field. Given a number $n$ such that $n < p^l$,
we can compute an element of $\F$ of order greater than $n$ in $\poly(p,l,n)$ time. 
Furthermore, we can find a primitive element in $\F$ in time  $|\F|^{\cO(1)}$. 
\end{lemma}
\begin{proof} We begin by applying Lemma~\ref{lemma: primitive set} to the field $\F$ and
obtain a set $S$ of size $\poly(p,l)$. This takes time $\poly(p,l)$. Then for each element $\alpha \in S$ we compute the set 
$G_\alpha = \{\alpha^i ~|~ i=1,2, \ldots ,n+1\}$.
If for any $\alpha$, $|G_\alpha| = n+1$, then we return this as the required element
of order greater than $n$. Since the set $S$ contains at least one primitive element of $\F$,
we will find some $\alpha$ in this step. Note this this step too takes $\poly(p,l,n)$ time. 

To prove the second statement of the lemma do as follows.  For each element $\alpha \in S$, consider the set 
$S(\alpha) = \{ \alpha^i ~|~ i=1, \ldots, p^l \}$.
  If $|S(\alpha)| = |\F^*| = p^l - 1$, then $\alpha$ generates $\F^*$. Since the set $S$ contains at least one primitive element of $\F$,
we will find some $\alpha$ in this step. Note this this step can be done in time $|\F|^{\cO(1)}$. 
This completes the proof of this lemma.
\end{proof}


When given a small field, the following lemma allows us to increase the size of the field
as well as find a primitive element in the larger field.
\begin{lemma} \label{lemma: increase field size}
Given a field $\F = \F_{p^l}$ and a number $n$ such that $p^l < n$,
we can find an extension $\K$ of $\F$ such that $n<|\K| < n^2$
and a primitive element $\alpha$ of $\K$ in time $n^{\cO(1)}$.
\end{lemma}
\begin{proof}

    Let $r$ be smallest number such that $p^{lr} > n$.
    But then $p^{lr/2} < n$. Therefore we have that $p^{lr} < n^2$.
    Next we find an extension of $\F$ of degree $r$,
    by finding an irreducible polynomial over $\F$ 
    of degree $r$ using Lemma~\ref{lemma: irred shoup}
    in time polynomial in $p, l, r$, which is $n^{\cO(1)}$.
    Then we can use Lemma~\ref{lemma: large order element}
    to compute a primitive element of $\K$.
    Since $|\K| < n^2$, this can be done in time $n^{\cO(1)}$. This completes the proof of this lemma.
\end{proof}

\subsection{Deterministic Truncation of Matrices}
In this section we look at algorithms for computing $k$-truncation of matrices.
We consider matrices  over the set of rational numbers ${\mathbb Q}$ or over some finite field $\F$.
Therefore, we are given as input a matrix $M$ of rank $n$ over a field $\F$.
Let $p$ be the characteristic of the field $\F$ and  $N$ denote the size of the input in bits.  
The following theorem, gives us an algorithm to compute the truncation of a matrix
using the classical wronskian, over an appropriate field.
We shall refer to this as the {\em classical wronskian method of truncation}.

\begin{lemma} \label{lemma: classic truncation}
    Let $M$ be a $n \times m$ matrix of rank $n$ over a field $\F$, where  $\F$ is either $\mathbb{Q}$ or $\ch(\F)>n$.  
    Then we can compute a $k \times m$ matrix $M_k$ of rank $k$ over the field $\F(X)$
    which is a $k$-truncation of the matrix $M$ in  $\cO(mnk)$ field operations.
\end{lemma}

\begin{proof} 
Let $\F[X]$ be the ring of polynomials in $X$ over $\F$ and let $\F(X)$ be the corresponding field of fractions.
Let $C_1,\ldots,C_m$ denote the columns of $M$.  
Observe that we have a polynomial $P_i(X)$ corresponding to the column $C_i$ of degree at most $n-1$,
and by Lemma~\ref{lemma: vectors and polynomials} we have that $C_{i_1}, \ldots, C_{i_l}$ are linearly independent over $\F$
if and only if $P_{i_1}(X), \ldots,P_{i_l}(X)$ are linearly independent over $\F$.
Further note that $P_i$ lies in $\F[X]$ and thus also in $\F(X)$. 
Let $D_i$ be the vector $(P_i(X), P_i^{(1)}(X),  \ldots, P_i^{(k-1)}(X))$  of length $k$
with entries from $\F[X]$ (and also in $\F(X)$).
Note that the entries of $D_i$ are  polynomials of degree at most $n-1$.
Let us define the matrix $M_k$ to be the $(k \times m)$ matrix whose columns are $D_i^T$, 
and note that $M_k$ is a matrix with entries from $\F[X]$.
We will show that indeed $M_k$ is a desired $k$-truncation of the matrix $M$ . 


Let $I\subseteq \{1,\ldots,m\}$ such that $|I| = l \leq k$.
Let $C_{i_1},\ldots,C_{i_l}$
be a linearly independent set of columns of the matrix $M$ over $\F$,  where $I = \{ i_1, \ldots, i_l \}$.
We will show that the columns $D_{i_1}^T, \ldots, D_{i_l}^T$ are linearly independent in $M_k$ over $\F(X)$.
Consider the $k \times l$ matrix $M_I$ whose column are the vectors $D_{i_1}^T, \ldots, D_{i_l}^T$.
We shall show that $M_I$ has rank $l$ by showing that there is a $l \times l$ submatrix 
whose determinant is a non-zero polynomial.
Let $P_{i_1}(X), \ldots, P_{i_l}(X)$ be the polynomials corresponding to the vectors $C_{i_1}, \ldots, C_{i_l}$.
By Lemma~\ref{lemma: vectors and polynomials} we have that $P_{i_1}(X),  \ldots, P_{i_l}(X)$ are linearly independent over $\F$.
Then by Theorem~\ref{thm: classic wronskian}, the $(l \times l)$ matrix formed by the column vectors 
$(P_{i_j}(X), P_{i_j}^{(1)}(X), \ldots, P_{i_j}^{(l-1)}(X))^T$, $i_j\in I$, is a non-zero determinant in $\F[X]$.
But note that this matrix is a submatrix of $M_I$. Therefore $M_I$ has rank $l$ in $\F(X)$.
Therefore the vectors $D_{i_1}^T, \ldots, D_{i_l}^T$ are linearly independent in $\F(X)$. This completes the proof of the forward direction. 

Let $I\subseteq \{1,\ldots,m\}$ such that $|I| = l \leq k$ and let $D_{i_1}^T, \ldots, D_{i_l}^T$ be 
linearly independent in $M_k$ over $\F(X)$, where $I = \{ i_1, \ldots, i_l \}$. We will show that the corresponding set of columns  $C_{i_1}, \ldots, C_{i_l}$ are also 
linearly independent over $\F$.  For a contradiction assume that  $C_{i_1}, \ldots, C_{i_l}$ are  linearly {\em dependent}  over $\F$. 
Let $P_{i_1}(X), \ldots, P_{i_l}(X)$ be the polynomials in $\F[X]$ corresponding to these vectors.  
Then by Lemma~\ref{lemma: vectors and polynomials} we have that $P_{i_1}(X), \ldots, P_{i_l}(X)$ are linearly dependent over $\F$.
So there is a tuple $a_{i_1}, \ldots, a_{i_l}$ of values of $\F$ such that $\sum_{j=1}^{l} a_{i_j} P_{i_j}(X) = 0$.
Therefore, for any $d \in \{1, \ldots, l-1 \}$, we have that $\sum_{j=1}^{l} a_{i_j} P_{i_j}^{(d)}(X) = 0$.
Now let $D_{i_1}^T, \ldots, D_{i_l}^T$ be the column vectors of $M_k$ corresponding to $C_{i_1}, \ldots, C_{i_l}$.
Note that $\F$ is a subfield of $\F(X)$ and by the above, we have that $\sum_{j=1}^{l} a_{i_j} D_{i_j} = 0$.
Thus $D_{i_1}^T, \ldots, D_{i_l}^T$ are linearly dependent in $M_k$ over $\F(X)$, a contradiction to our assumption.  

Thus we have shown that for any $\{i_1, \ldots, i_l\} \subseteq \{1, \ldots, m\}$ such that $l \leq k$,
$C_{i_1}, \ldots, C_{i_l}$ are linearly independent over $\F$ if and only if  $D_{i_1}, \ldots, D_{i_l}$ are linearly independent over 
$\F(X)$. To estimate the running time, observe that for each $C_i$ we can compute $D_i$ in $\cO(kn)$ field operations and thus we can compute $M_k$ in $\cO(mnk)$ field operations.  
This completes the proof of this lemma.
\end{proof}
Lemma~\ref{lemma: classic truncation} is useful in obtaining $k$-truncation of matrices which entries are either 
from the filed of large characteristic or from $\mathbb Q$.  
The following lemma allows us to find truncations in fields of small characteristic 
which have large order. We however require a primitive element or an element of high order 
of such a field to compute the truncation. In the next lemma we demand a lower bound on the size of the field as we need an element of certain order. We will later see how to 
remove this requirement from the statement of the next lemma.

%
\begin{lemma} \label{lemma: alpha folded truncation}
Let $\F$ be a finite field
and $\alpha$ be an element of $\F$ 
of order at least $(n-1)(k-1)+1$. 
Let $M$ be a $(n \times m)$ matrix of rank $n$ over a field $\F$.
Then we can compute a ($k \times m)$ matrix $M_k$ of rank $k$ over the field $\F(X)$ 
which is a $k$-truncation of the matrix $M$ in $\cO(mnk)$ field operations.
\end{lemma}
\begin{proof} 
Let $\F[X]$ be the ring of polynomials in $X$ over $\F$ and let $\F(X)$ be the corresponding field of fractions.
Let $C_1,\ldots,C_m$ denote the columns of $M$.  
Observe that we have a polynomial $P_i(X)$ corresponding to the column $C_i$ of degree at most $n-1$, 
and by Lemma~\ref{lemma: vectors and polynomials} we have that $C_{i_1}, \ldots, C_{i_l}$ are linearly independent over $\F$
if and only if $P_{i_1}(X), \ldots, P_{i_l}(X)$ are linearly independent over $\F$.
Further note that $P_i(X)$ lies in $\F[X]$ (and also in $\F(X)$). 
%
%

Let $D_i$ be the vector $(P_i(X), P_i( \alpha X),  \ldots, P_i(\alpha^{k-1}X))$.
Observe that the entries of $D_i$ are polynomials of degree at most $n-1$ and are elements of $\F[X]$.
Let us define the matrix $M_k$ to be the $(k \times m)$ matrix whose columns are the vectors $D_i^T$, 
and note that $M_k$ is a matrix with entries from $\F[X]\subseteq \F(X)$.
We will show that $M_k$ is a desired $k$-truncation of the matrix $M$ . 


Let $I\subseteq \{1,\ldots,m\}$ such that $|I| = l \leq k$.
Let $C_{i_1},\ldots,C_{i_l}$
be a  linearly independent set of columns of the matrix $M$ over $\F$,  where $I = \{ i_1, \ldots, i_l \}$.
We will show that $D_{i_1}^T, \ldots, D_{i_l}^T$ are linearly independent in $M_k$ over $\F(X)$.
Consider the $k \times l$ matrix $M_I$ whose columns are the vectors $D_{i_1}^T, \ldots, D_{i_l}^T$.
We shall show that $M_I$ has rank $l$ by showing that there is a $l \times l$ submatrix 
whose determinant is a non-zero polynomial.
Let $P_{i_1}(X), \ldots, P_{i_l}(X)$ be the polynomials corresponding to the vectors $C_{i_1}, \ldots, C_{i_l}$.  
By Lemma~\ref{lemma: vectors and polynomials} we have that $P_{i_1}(X),  \ldots, P_{i_l}(X)$ are linearly independent over $\F$. 
Then by Theorem~\ref{thm: alpha folded wronskian}, the $(l \times l)$ matrix formed by the column vectors 
$(P_{i_j}(X), P_{i_j}(\alpha X), \ldots, P_{i_j}(\alpha^{(l-1)}X))^T$, $i_j\in I$, is a non-zero determinant in $\F[X]$. 
But note that this matrix is a submatrix of $M_I$. Therefore $M_I$ has rank $l$ in $\F(X)$.
Therefore the vectors $D_{i_1}, \ldots, D_{i_l}$ are linearly independent in $\F(X)$. This completes the proof of the forward direction.

Let $I\subseteq \{1,\ldots,m\}$ such that $|I| = l \leq k$ and let $D_{i_1}^T, \ldots, D_{i_l}^T$ be 
linearly independent in $M_k$ over $\F(X)$, where $I = \{ i_1, \ldots, i_l \}$. We will show that the corresponding set of columns  $C_{i_1}, \ldots, C_{i_l}$ are also 
linearly independent over $\F$.  For a contradiction assume that  $C_{i_1}, \ldots, C_{i_l}$ are  linearly {\em dependent}  over $\F$. 
Let $P_{i_1}(X), \ldots, P_{i_l}(X)$ be the polynomials in $\F[X]$ corresponding to these vectors.  
Then by Lemma~\ref{lemma: vectors and polynomials} we have that $P_{i_1}(X), \ldots, P_{i_l}(X)$ are linearly dependent over $\F$
So there is a tuple $a_{i_1}, \ldots, a_{i_l}$ of values of $\F$ such that $\sum_{j=1}^{l} a_{i_j} P_{i_j}(X) = 0$.
Therefore, for any $d \in \{1, \ldots, l-1 \}$, we have that $\sum_{j=1}^{l} a_{i_j} P_{i_j}(\alpha^{d}X) = 0$.
Now let $D_{i_1}^T, \ldots, D_{i_l}^T$ be the column vectors of $M_k$ corresponding to $C_{i_1}, \ldots, C_{i_l}$.
Note that $\F$ is a subfield of $\F(X)$ and by the above, we have that $\sum_{j=1}^{l} a_{i_j} D_{i_j} = 0$.
Thus $D_{i_1}^T, \ldots, D_{i_l}^T$ are linearly dependent in $M_k$ over $\F(X)$, a contradiction to our assumption.

Thus we have shown that for any $\{i_1, \ldots, i_l\} \subseteq \{1, \ldots, m\}$ such that $l \leq k$,
$C_{i_1}, \ldots, C_{i_l}$ are linearly independent over $\F$ if and only if  $D_{i_1}, \ldots, D_{i_l}$ are linearly independent over 
$\F(X)$. To estimate the running time, observe that for each $C_i$ we can compute $D_i$ in $\cO(kn)$ field operations and thus we can compute $M_k$ in $\cO(mnk)$ field operations.  This completes the proof of this lemma.
%
\end{proof}

In Lemma~\ref{lemma: alpha folded truncation} we require that $\alpha$ be an element of order at least $(n-1)(k-1)+1$. This implies that 
the order of the field $\F$ must be at least $(n-1)(k-1)+1$.
We can ensure these requirements by preprocessing the input before invoking the Lemma~\ref{lemma: alpha folded truncation}. 
Formally, we show the following lemma.

\begin{lemma} \label{lemma: folded preprocess}
Let $M$ be a matrix of dimension $n \times m$ over a finite field $\F$, and of rank  $n$.
Let $\F = \F_{p^\ell}$ where $p < n$. Then in polynomial time we can find an extension field $\K$ 
of order at least $nk + 1$ 
and an element $\alpha$ of $\K$ of order at least $nk+1$,
such that $M$ is a matrix over $\K$ with the same linear independence relationships
between it's columns as before.
\end{lemma}

\begin{proof}  We distinguish two cases by comparing the values of $p^l$ and $n$.
\begin{description}
\item [Case 1: $p^\ell \leq nk+1$. ] In this case we use Lemma~\ref{lemma: increase field size} to obtain 
an extension $\K$ of $\F$ of  size at most $(nk+1)^2$, and a primitive element $\alpha$ of $\K$
in polynomial time. 
\item[Case 2: $nk+1 < p^\ell$.]  In this case we set $\K = \F$ and and use Lemma~\ref{lemma: large order element} 
to find an element of order at least $nk$, in time $\poly(p,l,nk)$.
\end{description}
%
%

Observe that  $\F$ is a subfield of $\K$, $M$ is also a matrix over $\K$. Thus, 
any set of linearly dependent columns over $\F$ continue to be linearly dependent over $\K$.
Similarly, linearly independent columns continue to be linearly independent.  
This completes the proof of this lemma.
\end{proof}

Next we show a result that allows us to find basis of matrices with entries from $\F[X]$. 
\begin{lemma}
\label{lemma:ind}
 Let $M$ be a  $m \times t$ matrix with entries from $\F[X]^{<n}$ and let $m\leq t$.
Let $w~:~\mathbf{C}(M)\rightarrow {\mathbb R}^+$ be a weight function. 
Then we can compute minimum weight column basis of $M$ in $\cO(m^2n^2t+m^{\omega}nt)$ field operations in $\F$
\end{lemma}
\begin{proof}
Let $S\subseteq \F^*$ be a set of size $(n-1)m+1$ and for every $\alpha \in S$, 
let $M(\alpha)$ be the matrix obtained by substituting $\alpha$ for $X$ in the polynomials in matrix $M$. 
Now we compute the minimum weight column basis $C(\alpha)$ in $M(\alpha)$ for all $\alpha\in S$. 
Let $l=\max\{|C(\alpha)|~|~\alpha\in S\}$. Among all the column basis of size $l$, let $C(\zeta)$ be a minimum weighted column basis for some $\zeta \in S$.  
Let $C'$ be the columns in $M$  corresponding to $C(\zeta)$. 
We will prove that $C'$ is a minimum weighted column basis of $M$.  Towards this we start with the following claim. 

 \begin{claim}
\label{claim:rnakofmiandm}
 The rank of $M$ is the maximum of the rank of matrices $M(\alpha)$, $\alpha \in S$. 
 \end{claim}
 \begin{proof}
 Let $r\leq m$ be the rank of $M$. Thus, we know that there exists a submatrix $W$of $M$  of dimension $r\times r$ such that 
 $\mathrm{det}(W)$ is a non-zero polynomial.   The degree of  the polynomial  $\mathrm{det}(W(X))\leq (n-1)\times r \leq (n-1) m$. Thus, we know that it has at-most $(n-1)m$ roots. Hence, when we evaluate $\mathrm{det}(W(X))$ on  set $S$ of size more than $(n-1)m$,  
 there exists at least one element in $S$, say $\beta$, such that  $\mathrm{det}(W(\beta))\neq 0$. Thus, the rank of 
 $M$ is upper bounded by the rank of  $M(\beta)$ and  hence upper bounded by the   
 maximum of the rank of matrices $M(\alpha)$, $\alpha \in S$. 
 
 As before let $r\leq l$ be the rank of $M$.  
 Let $\alpha$ be an arbitrary element of $S$. Observe that for any submatrix $Z$ of dimension $r'\times r'$, $r'>r$  we have that 
 $\mathrm{det}(Z(X))\equiv 0$. Thus, for any $\alpha$, the determinant of the corresponding submatrix of $M(\alpha)$ is also $0$. This implies that for any $\alpha$, the rank of $M(\alpha)$ is at most $r$. 
 This completes the proof.  
 \end{proof}
 \noindent 
Claim~\ref{claim:rnakofmiandm} implies that $l=\max\{|C(\alpha)|~|~\alpha\in S\}$  is equal to the rank of $M$.  Our next claim is following. 
\begin{claim}
\label{claim:ind}
For any $\alpha \in S$,  and $C\subseteq \mathbf{C}(M(\alpha))$, if $C$ is  linearly independent in $M(\alpha)$ then $C$ is also linearly independent in $M$.
\end{claim}
The proof follows from the arguments similar to the ones used in proving reverse direction of Claim~\ref{claim:rnakofmiandm}. 
Let $r\leq m$ be the rank of $M$ and let $C^*$ be a minimum weight column basis of $M$. 
Thus, we know that there exists a submatrix $W$ of $M$  of dimension $r\times r$ such that 
$\mathrm{det}(W)$ is a non-zero polynomial.  The degree of  the polynomial  $\mathrm{det}(W(X))\leq (n-1)\times r \leq (n-1) m$. 
Thus, we know that it has at most $(n-1)r$ roots. Hence, when we evaluate $\mathrm{det}(W(X))$ on  set $S$ of size more than 
$(n-1)r$, there exists at least one element in $S$, say $\beta$, such that  $\mathrm{det}(W(\beta))\neq 0$ and the set 
of columns $C^*$ is linearly independent in $M(\beta)$.  Using Claim~\ref{claim:ind} and the fact that $C^*$ is linearly independent 
in both $M(\beta)$ and $M$, we can conclude that $C^*$ is a column basis for $M(\beta)$. Since $|C'|=|C^*|$, $w(C')\leq w(C^*)$, 
$C'$ is indeed  a minimum weighted column basis of $M$.

%

We can obtain any $M(\alpha)$ with at most $\cO(nmt)$ field operations in $\F$. 
Furthermore, we can compute  minimum weight column basis of $M(\alpha)$ in  $\cO(tm^{\omega-1})$ field operations~\cite{BodlaenderCKN13}.  
Hence the total number of field operations over $\F$ is bounded by  $\cO(m^2n^2t+m^{\omega}nt)$.
%
%
\end{proof}

Finally, we combine Lemma~\ref{lemma: classic truncation},  Lemma~\ref{lemma: folded preprocess} and Lemma~\ref{lemma: alpha folded truncation}
to obtain the following theorem.

\begin{theorem}[Theorem~\ref{main:thm: truncation}, restated]
 \label{thm: truncation}
Let $\F=\F_{p^l}$ be a finite field or $\F=\mathbb{Q}$. 
Let $M$ be a $n\times m$ matrix over $\F$ of rank $n$. 
Given a number $k \leq n$, we can compute a matrix $M_k$ over the field $\F(X)$
such that it is a representation of the $k$-truncation of $M$.   
Furthermore, given $M_k$, we can test whether a given set of $l$ columns 
in $M_k$ are linearly  independent in $\cO(n^2k^3)$ field operations.   
\end{theorem}
\begin{proof}
We first consider the case when the characteristic of the field $\F$ is $0$,
or if the characteristic  $p$ of the finite field is strictly larger than $n$ ($\ch(\F)>n$). 
In this case we apply Lemma~\ref{lemma: classic truncation} to obtain a matrix $M_k$ over $\F(X)$
which is a $k$-truncation of $M$. 
We now  consider the case when $\F$ is a finite field and $\ch(\F) <n$. 
First apply Lemma~\ref{lemma: folded preprocess} to ensure that
the order of the field $\F$ is greater than $(n-1)(k-1) + 1$ and to
obtain an element of order at least $(n-1)(k-1) + 1$  in the field $\F$.  Of course by doing this, 
we have gone to an extension of $\K$ of $\F$  of size at least $(n-1)(k-1)+1$. However, for brevity of presentation 
we will assume that the input is given over such an extension. 
We then apply Lemma~\ref{lemma: alpha folded truncation}
to obtain a matrix $M_k$ over $\F(X)$ which is a representation
of the $k$-truncation of the matrix $M$.  One should notice that in fact  
This completes the description of $M_k$ in all the cases. 

Let $I\subseteq \{1,\ldots,m\}$ such that $|I| = l \leq k$.
Let $D_{i_1},\ldots,D_{i_l}$
be a set of columns of the matrix $M_k$ over $\F$,  where $I = \{ i_1, \ldots, i_l \}$. 
Furthermore, by $M_I$ we denote the $k\times \ell$ 
submatrix of $M_k$ containing the columns   $D_{i_1},\ldots,D_{i_l}$.  
To test whether these columns are linearly independent, we can apply Lemma~\ref{lemma:ind} on $M_I^T$ and see the 
size of column basis of $M_I^T$ is $l$ or not. This takes time $\cO(l^2n^2k+l^\omega nk)=\cO(n^2k^3)$ field operations in $\F$.   
\end{proof}

\subsubsection{Representating the truncation over a finite field}
In Theorem~\ref{thm: truncation}, the representation $M_k$ is over the field $\F(X)$.
However, in some cases  this matrix can also be viewed as a representation over a finite extension of $\F$ of sufficiently 
large degree. That is, if $\F=\F_{p^l}$ is a finite field then $M_k$ can be given over $\F_{p^{l'}}$ where $l'\geq nkl$. 
Formally we have the following lemma.

\begin{theorem} \label{lemma: represent truncation over finite field}
Let $M$ be a $n\times m$ matrix over $\F$ of rank $n$, $k\leq n$ be a positive integer and $N$ be the size of the input matrix. If $\F=\F_{p}$ be a prime field or $\F=\F_{p^l} $
where $p = N^{\cO(1)}$, 
then in polynomial time we can find a $k$-truncation $M_k$
of $M$  over a finite extension $\K$ of $\F$ where $\K=\F_{p^{nkl}}$. 
\end{theorem}

\begin{proof}  Let $M_k$ be the matrix returned by Theorem~\ref{thm: truncation}.  Next we show how we can view the entries in $M_k$ over a finite extension of $\F$. 
Consider any extension $\K$ of $\F$ of degree $r \geq nk$.
Thus $\K = \frac{\F[X]}{r(X)}$, where 
$r(X)$ is a irreducible polynomial
in $\F[X]$ of degree $r$.  Recall that each entry of $M_k$ is a polynomial in $\F[X]$ of degree
at most $n-1$ and therefore they are present in $\K$.
Further the determinant of any $k \times k$ submatrix of $M_k$
is identically zero in $\K$ if and only if it is identically zero in $\F(X)$.
This follows from the fact that the determinant is a polynomial of degree
at most $(n-1)k$ and therefore is also present in $\K$. Thus $M_k$ is a representation over $\K$.

To specify the field $\K$ we need to compute the irreducible polynomial $r(X)$.
If $\F$ is a prime field, i.e. $\F = \F_p$, then we can compute the polynomial $r(X)$
using the first part of Lemma~\ref{lemma: irred shoup}. 
And if $p = N^{\cO(1)}$ we can use the second part of  Lemma~\ref{lemma: irred shoup} to compute $r(X)$.
Thus we have a well defined $k$-truncation of $M$ over the finite field $\K = \frac{\F[X]}{r(X)}$. Furthermore, if degree of $r(X)$ 
is $nk$ then $\K$ is isomorphic to  $\F_{p^{nkl}}$. 
This completes the proof of this theorem. 
\end{proof}
\section{ Application to Computation of Representative Families}
In this section we give deterministic algorithms to compute representative families of a linear matroid, given its representation matrix. 
We start with the definition of a {\em $q$-representative family}.  
\begin{definition}[{\bf $q$-Representative Family}]
Given a matroid  \mat{} and a family $\cal S$ of subsets of $E$, we say that a subfamily $\widehat{\cal{S}}\subseteq \cal S$ 
is {\em $q$-representative} for $\cal S$ 
if the following holds: for every set $Y\subseteq  E$ of size at most $q$, if there is a set $X \in \cal S$ disjoint from $Y$ with $X\cup Y \in \I$, then there is a set $\whnd{X} \in \whnd{\cal S}$ disjoint from $Y$ with $\whnd{X} \cup  Y \in \I$.  If $\hat{\cal S} \subseteq {\cal S}$ is $q$-representative for ${\cal S}$ we write \rep{S}{q}. 
\end{definition}

In other words if some independent set in $\cal S$ can be extended to a larger independent set by $q$ new elements, then there is a set in 
$\widehat{\cal S}$ that can be extended by the same $q$ elements.   We say that a family  $ \cS = \{S_1,\ldots, S_t\}$ of 
sets is a {\em $p$-family} if each set in $\cal S $ is of size $p$. In~\cite{FominL14abs-1304-4626,FominLS14} the following theorem is proved. See \cite[Theorem~4]{FominL14abs-1304-4626}.


%

\begin{theorem}[\cite{FominL14abs-1304-4626,FominLS14}]
\label{thm:repsetlovaszrandomized}
Let \mat{}   be a linear matroid and let $ \cS = \{S_1,\ldots, S_t\}$ be a $p$-family of independent sets. Then there 
exists \rep{S}{q} of size \bnoml{p+q}{p}. 
Furthermore, given a representation \repmat{M}  of $M$ over a field $ \mathbb{F}$, there is a randomized algorithm  computing  
\rep{S}{q} in  \tgem \, operations over $ \mathbb{F}$. 
 \end{theorem}
Let $p+q=k$.  Fomin et al.~\cite[Theorem 3.1]{FominLS14} first  gave  a deterministic algorithm for computing $q$-representative of a $p$-family of indpendent sets if the rank of the corresponding matroid is $p+q$.  To prove Theorem~\ref{thm:repsetlovaszrandomized} one first computes the representation matrix of a 
$k$-truncation of \mat. 
This step returns a representation of a $k$-truncation of \mat{} with a high probability.  Given this matrix, one 
applies \cite[Theorem 3.1]{FominLS14} and arrive at Theorem~\ref{thm:repsetlovaszrandomized}.  In this section we design a deterministic 
algorithm for computing $q$-representative even if the underlying linear matroid has unbounded rank, using deterministic truncation of linear matroids.


Observe that the representation given by Theorem~\ref{thm: truncation}  is over $\F(X)$. For the purpose of computing $q$-representative of a $p$-family of independent sets we need to find a set of linearly independent columns over a matrix withe entries from $\F[X]$.  However, deterministic   algorithms to compute basis of matrices over $\F[X]$  is not as fast as compare to the algorithms where we do not need to do symbolic computation. We start with a lemma that allows us to find  a set of linearly independent columns of a matrix over $\F[X]$ quickly; though the size of the set returned by the algorithm given by the lemma  could be slightly larger than the basis of the given matrix. 
\begin{define}
Let $W=\{v_1,\ldots,v_m\}$ be a set of vectors  over $\F$ and $w:W \rightarrow {\mathbb R}^{+}$. We say that $S\subseteq W$ is a 
{\em spanning set}, if every $v\in W$ can be written as linear combination of vectors in  $S$ with coefficients 
from $\F$. We say that $S$ is a {\em nice spanning set of $W$}, if $S$ is a spanning set and for  any $z\in W$
if $z = \sum_{v \in S} \lambda_v v$, and $ 0 \neq \lambda_v \in \F$ then we have $w(v) \leq w(z)$. 
\end{define}
The following lemma  enables us to find a small size 
spanning set  of vectors over $\F(X)$.

\begin{lemma} \label{lemma:GE in Fx}
Let $\F$ be a field and 
let $M \in \F[X]^{m \times t}$ be a matrix over $\F[X]^{<n}$ 
and let $w:\mathbf{C}(M) \rightarrow {\mathbb R}^{+}$ be a weight function. 
Then we can find a nice spanning set $S$  of $\mathbf{C}(M) $
of size at most $nm$ with at most 
$\cO(t (nm)^{\omega -1})$ field operations. 
\end{lemma}

\begin{proof} 
The main idea is to do a ``gaussian elimination''  in $M$, but only over the subfield $\F$ of $\F(X)$.
    Let $C_i$ be a column of the matrix $M$. 
    It is a vector of length $m$ over $\F[X]^{<n}$ and it's entries are polynomials $P_{ji}(X)$, where $j \in \{1, \ldots ,m\}$.
    Observe that $P_{ji}(X)$ is a polynomial of degree $n-1$ with coefficients from $\F$. 
    Let $v_{ji}$ denote the vector of length $n$  corresponding to the polynomial $P_{ji}(X)$. 
    Consider the column vector $v_i$ formed by concatenating each $v_{ji}$ in order from $j=1$ to $m$. 
    That is, $v_i=(v_{1i},\ldots,v_{mi})^T$.
    This vector has length $nm$ and has entries from $\F$.
    Let $N$ be the matrix where columns correspond to column vectors $v_i$.
    Note that $N$ is a matrix over $\F$ of dimension $ nm \times t$ and
    the time taken to compute $N$ is $\cO(tnm)$.
    For each column $v_i$ of $N$ we define it's weight to be $w(C_i)$.
    We now do a gaussian elimination in $N$ over $\F$ and compute a minimum weight set of column vectors $S'$, which spans $N$.
    Observe that $|S'| \leq nm$ and time taken to compute $S'$ is $\cO(t (nm)^{\omega -1 })$~\cite{BodlaenderCKN13}. 
    Let $S$ be the set of column vectors in $M$ corresponding to  the column vectors in $S'$.
    We return $S$ as a nice spanning set of column vectors in $M$. 

    Now we show the correctness of the above algorithm.
    We first show that $S$ is a spanning set of $M$.
    Let $v_1, \ldots, v_{|S|}$ be the set of vectors in $S$ and let $v_d$ be some column vector in $N$. 
    Then $v_d = \sum_{i=1}^{|S|} a_i v_i$ where $a_i \in \F$.
    In particular for any $j \in \{1, \ldots, m\}$ we have $v_{jd} = \sum_{i=1}^{|S|} a_i v_{ji}$.
    Let $C_1, \ldots, C_{|S|}$ be the column vectors corresponding to $v_1, \ldots, v_{|S|}$
    and let $C_d$ be the column vector corresponding to $v_d$.
    We claim that $C_d = \sum_{i=1}^{|S|} a_i C_i$.
    Consider the $j$-the entry of the column vector $C$ and of $C_1, \ldots, C_{|S|}$. Towards our claim 
    we need to show that  $P_{jd}(X) = \sum_{i=1}^{|S|} a_i P_{ji}(X)$.
    But since $v_{dj}$ and $\{v_{ij}~|~j\in \{1,\ldots,m\}\}$ are the collection of vectors corresponding to $P_{jd}(X)$ and 
    $\{P_{ji}(X)~|~j\in \{1,\ldots,m\}\}$, the claim follows.

Next we show that $S$ is indeed a nice spanning set. Since $S$ is a spanning set of $M$ we have that any column 
$C_d = \sum_{C_i \in S} \lambda_i C_i$, $\lambda_i \in \F$.
    Let $C_j \in S$ be such that $\lambda_j \neq 0$ and $w(C_j) > w(C_d)$.
    Let $v_d$ and $v_j$ be the vectors corresponding to $C_d$ and $C_j$ respectively.
    We have that $v_d = \sum_{v_i \in S} \lambda_i v_i$, which implies 
    $v_j = \lambda_j^{-1}v_d - \sum_{v_i \in S, v_i \neq v_j} \lambda_j^{-1} \lambda_i v_i $.
    But this implies that $S^* = (S \setminus \{v_j\}) \cup \{v_d\}$ is a spanning set of $N$,
    and $w(S^*) < w(S)$, which is  a contradiction.
    Thus we have that for every column vector $C \in M$
    if $C = \sum_{C_i \in S} \lambda_i C_i$ and $ 0 \neq \lambda_i \in \F$, then $w(C_i) \leq w(C)$. This completes the proof.
\end{proof}

  %
  %
  %

The main theorem of this section is as follows. 

%
%
%

\begin{theorem}[Theorem~\ref{main:thm:repsetlovaszrandomized}, restated]
\label{thm:repsetlovaszrandomized}
Let \mat{}   be a linear matroid  of rank $n$ and let $ \cS = \{S_1,\ldots, S_t\}$ be a $p$-family of independent sets. 
Let $A$ be a $n\times |E|$ matrix representing $M$ over a field $ \mathbb{F}$, where $\F=\F_{p^\ell}$ 
or  $\F$ is $\mathbb Q$. Then there  are deterministic algorithms  computing  
\rep{S}{q}  as follows. 
\begin{enumerate}
\setlength{\itemsep}{-2pt}
\item A family $\widehat{\cal S}$ of size \bnoml{p+q}{p} in  $\cO\left({p+q \choose p}^2 t p^3n^2 + t {p+q \choose q} ^{\omega} np\right)+(n+|E|)^{\cO(1)}$, operations over $ \mathbb{F}$. 
\item A family   $\widehat{\cal S}$ of size  $ np {p+q \choose p}$ in  $\cO\left({p+q \choose p} t p^3n^2 
+ t {p+q \choose q}^{\omega-1} (pn)^{\omega-1} \right)+(n+|E|)^{\cO(1)}$ operations over $ \mathbb{F}$.  
\end{enumerate}
 \end{theorem}
\begin{proof}
Let $p+q=k$ and $|E|=m$. We start by finding $k$-truncation of $A$, say $A_k$, over $\F[X]\subseteq \F(X)$ using Theorem~\ref{thm: truncation}.  We can find $A_k$ with at most $(n+m)^{\cO(1)}$  operations over $ \F$.  Given the matrix $A_k$ we follow the proof of 
\cite[Theorem 3.1]{FominLS14}. 
For a set $S\in \cal S$ and $I \in {[k]\choose p}$, we define   $s[I]=\det(A_k[I,S])$.  
We also define 
\[\vec{s}_i=\left(s_i[I]\right)_{I\in {[k]\choose p} }.\]
Thus the entries of the vector $\vec{s}_i$ are the values of $\det(A_k[I,S_i])$, where $I$ runs through all the $p$ sized subsets of rows of $A_k$.   
Let $H_\mathcal{S}=(\vec{s}_1,\ldots, \vec{s}_t)$ be the ${k \choose p} \times t$ matrix obtained by taking $\vec{s}_i$ as columns.  
Observe that each entry in $A_k$ is in $\F[X]^{<n}$.  Thus, the determinant polynomial corresponding to any $p\times p$ submatrix of $A_k$ has degree at most $pn$. It is well known that we can find determinant of a $p\times p$ matrix over   $\F[X]^{<n}$ in time $\cO(p^3n^2)$~\cite{MuldersS03}.  
Thus, we can obtain  $H_\mathcal{S}$ in time $\cO(t{p+q \choose p}p^3n^2)$. 

Let $W$ be a spanning set of columns for $\mathbf{C}(H_\mathcal{S})$.   We define 
$\widehat{W}=\{S_\alpha~|~\vec{s}_\alpha \in W \}$ as the corresponding  subfamily of  $\cal S$ . 
The proof of ~\cite[Theorem 3.1]{FominLS14} implies that if $W$ is a spanning set of columns for $\mathbf{C}(H_\mathcal{S})$ then the corresponding $\widehat{W}$  is the required $q$-representative family for $\cal S$. That is, $\widehat{W}\subseteq_{rep}^q {\cal S}$.  We get the desired running time by either using Lemma~\ref{lemma:ind} to compute a basis of size ${p+q\choose p}$ for $H_\mathcal{S}$  or by using Lemma~\ref{lemma:GE in Fx} to compute a spanning set of size  $np{p+q\choose p}$ of $\mathbf{C}(H_\mathcal{S})$.  This completes the proof. 
\end{proof}

In fact one can prove Theorem~\ref{thm:repsetlovaszrandomized} for a ``weighted notion of representative family'' (see~\cite{FominLS14} for more details), for which we would need to compute nice spanning  set.  It will appear in an extended version of the paper. 

\subsection{Applications}

 Marx~\cite{Marx09} gave algorithms for several problems  based on matroid optimization. The main theorem in his work is 
 Theorem 1.1~\cite{Marx09} on which most applications of~\cite{Marx09} are based. The proof of the theorem uses an algorithm to find representative sets as a black box. Applying our algorithm (Theorem~\ref{thm:repsetlovaszrandomized} of this paper) instead gives a 
deterministic  version of 
 Theorem 1.1 of \cite{Marx09}. 

\begin{proposition}
\label{prop:marxmainresult}
Let \mat{} be a linear matroid where the ground set is partitioned into blocks of size $\ell$. Given a linear representation $A_M$ of $M$, it can be determined in $\cO(2^{\omega k\ell} ||A_M||^{\cO(1)})$  time whether there is an independent set that is the union of $k$ blocks. ($||A_M||$ denotes 
the length of $A_M$ in the input.)
\end{proposition}

Finally, we mention another application from~\cite{Marx09} which we believe could be useful to obtain single exponential time 
parameterized and exact algorithms.  

 \medskip
\begin{center} 
\fbox{\begin{minipage}{0.96\textwidth}
\noindent{\sc  $\ell$-Matroid Intersection} \hfill {\bf Parameter:} $k$ \\
\noindent {\bf Input}: Let $M_1=(E,{\cal I}_1),\dots, M_1=(E,{\cal I}_\ell) $ be matroids on the same ground set $E$  given \\
 \noindent{\phantom{{\em Input}:}} by their representations $A_{M_1},\ldots, A_{M_{\ell}}$ over the same field $\mathbb{F}$ and a positive integer $k$.\\
\noindent{\bf Question}: Does there exist $k$ element set that is independent in each $M_i$ ($X\in {\cal I}_1\cap \ldots  \cap {\cal I}_\ell$)?  
\end{minipage}}
\end{center}
\medskip
 
 Using Theorem 1.1 of ~\cite{Marx09}, Marx~\cite{Marx09} gave a randomized algorithm for {\sc  $\ell$-Matroid Intersection}. By using 
Proposition~\ref{prop:marxmainresult} instead we get  the following result. 
\begin{proposition}
{\sc  $\ell$-Matroid Intersection} can be solved in $\cO(2^{\omega k\ell} ||A_M||^{\cO(1)})$  time. 
\end{proposition}

\section{Conclusions}\label{sec:conc}
In this paper we give the first deterministic algorithm to compute a $k$-truncation of linear matroids over 
any finite field or the field $\mathbb Q$. Our algorithms were based on the properties of the Wronskian determinant 
and the $\alpha$-folded Wronskian determinant, where $\alpha$ is an element of polynomially large order in the field $\F$.
We believe that our results on the $\alpha$-folded Wronskian will be useful in other contexts.
We conclude with a few related open problems.
\begin{itemize}
\item Our algorithm produces a representation of the truncation over the ring $\F[X]$ when the input field is $\F$.
	  However when $\F$ is large enough, then one can obtain a randomized representation of the truncation over $\F$ itself.
	  It is an interesting problem to compute the representation over $\F$ deterministically.

\item In many cases, when the field $\F$ is nice enough, we can represent the truncation over a finite degree extension 
          of the input field $\F$. It would be interesting to extend this to all finite field.

\item Finally, finding a deterministic representation of Transversal matroids and Gammoids, remain an interesting open problem
         in Matroid Theory. A solution to this problem will lead to a deterministic kernelization algorithm for several improtant graph problems in Parameterized Complexity \cite{KratschW12, kratsch2012compression}.

\end{itemize}

\newpage
\bibliographystyle{siam}
\bibliography{det-trunc.bib}

\end{document}